\documentclass[11pt, a4paper, onecolumn, accepted=2021-12-23]{quantumarticle}
\pdfoutput=1
\usepackage[utf8]{inputenc}
\usepackage[english]{babel}
\usepackage[T1]{fontenc}
\usepackage{amsmath}
\usepackage{hyperref}

\usepackage{physics}
\usepackage{braket}
\usepackage{dsfont}
\usepackage{amssymb}
\usepackage{amsthm}
\usepackage{fontawesome}
\usepackage{stmaryrd}

\newcommand{\jewel}{\text{\faDiamond}}

\renewcommand{\epsilon}{\varepsilon}
\renewcommand{\phi}{\varphi}

\newmuskip\pFqmuskip

\newcommand*\pFq[6][8]{
  \begingroup
  \pFqmuskip=#1mu\relax
  \mathchardef\normalcomma=\mathcode`,
  \mathcode`\,=\string"8000
  \begingroup\lccode`\~=`\,
  \lowercase{\endgroup\let~}\pFqcomma
  {}_{#2}F_{#3}{\left[\genfrac..{0pt}{}{#4}{#5};#6\right]}
  \endgroup
}
\newcommand{\pFqcomma}{{\normalcomma}\mskip\pFqmuskip}

\newtheorem{thm}{Theorem}[section]
\newtheorem{lem}[thm]{Lemma}
\newtheorem{cor}[thm]{Corollary}

\newtheorem{defi}[thm]{Definition}
\newtheorem{prop}[thm]{Proposition}
\newtheorem{remark}[thm]{Remark}

\usepackage[numbers,sort&compress]{natbib}

\begin{document}

\title{Maximal violation of steering inequalities and the matrix cube}

\author{Andreas Bluhm}
\email{bluhm@math.ku.dk}
\orcid{0000-0003-4796-7633}
\affiliation{QMATH, Department of Mathematical Sciences, University of Copenhagen, Denmark}

\author{Ion Nechita}
\email{nechita@irsamc.ups-tlse.fr}
\affiliation{Laboratoire de Physique Th\'eorique, Universit\'e de Toulouse, CNRS, UPS, France}
\orcid{0000-0003-3016-7795}
\maketitle

\begin{abstract}
In this work, we characterize the amount of steerability present in quantum theory by connecting the maximal violation of a steering inequality to an inclusion problem of free spectrahedra. In particular, we show that the maximal violation of an arbitrary unbiased dichotomic steering inequality is given by the inclusion constants of the matrix cube, which is a well-studied object in convex optimization theory. This allows us to find new upper bounds on the maximal violation of steering inequalities and to show that previously obtained violations are optimal. In order to do this, we prove lower bounds on the inclusion constants of the complex matrix cube, which might be of independent interest. Finally, we show that the inclusion constants of the matrix cube and the matrix diamond are the same. This allows us to derive new bounds on the amount of incompatibility available in dichotomic quantum measurements in fixed dimension.
\end{abstract}

\tableofcontents

\section{Introduction}

Quantum steering is a phenomenon which was already discovered in the early days of quantum mechanics. Schr{\"o}dinger noticed in 1936 \cite{Schroedinger1935, Schroedinger1936} that in a bipartite setting, one party can steer the state of the other party only using local measurements. Quantum steering was formalized in 2007 \cite{Wiseman2007, Jones2007} and given the operational interpretation that a party Alice wants to convince another party, Bob, that both share an entangled quantum state. The setup is such that Bob does not trust Alice, but he trusts his own devices. Quantum steering has since received a lot of attention due to its applications to, among other things, one-sided device-independent quantum
key distribution \cite{Branciard2012, He2013} and sub-channel discrimination \cite{Piani2015}. 

Already in \cite{Wiseman2007, Jones2007}, it has been shown that the correlations needed for quantum states to exhibit steering are stronger than mere entanglement, but weaker than the correlations required to violate a Bell inequality (see \cite{Uola2020} for a topical review of quantum steering). Therefore, in the same spirit as for Bell inequalities, we can consider \emph{steering inequalities} \cite{Cavalcanti2009} and compare their classical value (more precisely the value obtained using \emph{local hidden states} (LHS) models) to their quantum value. These are the objects of study in \cite{marciniak2015unbounded, Yin2015}. The authors show that increasing the dimension of the quantum systems, the violation of certain steering inequalities can become arbitrarily large. Experimental demonstrations of this fact was proposed in \cite{Skrzypczyk2015, Rutkowski2017}. Our work continues this line of investigation by making the connection to inclusions of free spectrahedra, in particular the \emph{matrix cube}, which is a well studied object in convex optimization theory \cite{ben-tal2002tractable}. Similar connections have been established previously by the authors of this work in the context of the compatibility of quantum measurements \cite{bluhm2018joint, bluhm2020compatibility, Bluhm2020GPT}. This connection allows us to find tight upper bounds on the violation of steering inequalities for ensembles of fixed dimension and generated by a fixed number of dichotomic measurements on Alice's side.

Our results are about the inclusion of the set of quantum assemblages having local hidden state models inside the set of all quantum assemblages, in a fixed setting, where the number of elements of the assemblages is fixed, and also the Hilbert space dimension is fixed. Informally, we show that the following three ``inclusion problems'' are equivalent: 
\begin{itemize}
    \item For an arbitrary steering inequality, how much smaller is its LHS value than its quantum value?
    \item How much does one has to shrink the set of quantum assemblages in order to contain it in the set of LHS assemblages?
    \item How much does one has to shrink the matrix cube in order to contain it in any free spectrahedron with the property that its level one contains the hypercube?
\end{itemize}
One can recognize that the first two statements are dual. The main insight from our work is that the third problem, a well studied question in optimization theory, is equivalent to the first two. Using precise estimates for the inclusion constants of the matrix cube, we can derive, for the first time, bounds for the constants in the first two cases, which are, in some situations, tight. Importantly, the bounds we provide are not steering inequality- or quantum assemblage-dependent; they can be understood as worse-case bounds, independent of the particular instance of the steering inequality or quantum assemblage considered. These quantities are thus \emph{universal constants, measuring the amount of steerability present in quantum mechanics}, in a given steering scenario (assemblage ``shape'' and Hilbert space dimension).

For example, we show in Section \ref{sec:g-indep-LB} that \emph{any unbiased dichotomic steering inequality $\mathbf F$ in dimension $d$} has violation upper bounded by 
$$\frac{V_{\mathcal Q}(\mathbf{F})}{V_{\mathcal L}(\mathbf{F})} \leq \frac{4^{n}}{\binom{2n}{n}},$$
where $n=\lfloor d/2 \rfloor$. The value above is achieved by a sequence of steering inequalities obtained by discretizing the Haar measure on the unitary group $\mathcal U(d)$. This type of result is reminiscent of \cite{Designolle2021}, where the certification of genuine high-dimensional steering was considered, building on earlier work on \emph{dimension witnesses} \cite{brunner2008testing, Cong2017} in the framework of Bell inequalities. In other words, observing a violation of a steering inequality $\mathbf F$ larger than the value above guarantees that the Hilbert space dimension of the quantum states is larger than $d$. 

The bounds we compute also have implications for the incompatibility of dichotomic measurements: they improve the bounds on the compatibility region (the maximal amount of compatibility for fixed dimension and number of measurements) found in \cite{bluhm2018joint}.

After introducing the necessary background from free spectrahedra, measurement compatibility and quantum steering in Section \ref{sec:prelim}, we show the connection between free spectrahedral inclusion and the values of steering inequalities in Section \ref{sec:spec-form}. In Section \ref{sec:incl-csts}, we explore this connection further and show that the sets of inclusion constants for the matrix cube has a correspondence on the steering side, which we call steering constants. The steering constants quantify the maximal violation of steering inequalities. We find another identification of the inclusion constants with the robustness of quantum steering in Section \ref{sec:relative-size}. Section \ref{sec:cube-csts} is devoted to proving bounds on the inclusion constants for the matrix cube. Finally, we conclude with a comparison of our results to previous work in both quantum steering and measurement incompatibility in Section \ref{sec:discussion}.

\section{Preliminaries} \label{sec:prelim}

\subsection{Notation}

Let us write $[n]:= \{1, \ldots, n \}$ for $n \in \mathbb N$. Moreover, for $d \in \mathbb N$, let $\mathcal M_d$ be the set of $d \times d$ matrices with complex entries. For the positive semidefinite matrices, we write $\mathcal M_d^+$, and we denote by ``$\leq$'' the positive semidefinite order on matrices: $X \leq Y \iff Y-X \in \mathcal M_d^+$. The self-adjoint matrices will be denoted by $\mathcal M_d^{\mathrm{sa}}$. The identity matrix in dimension $d$ is written as $I_d$, where we will drop the index if there is no risk of confusion. We denote by $\{e_i\}_{i \in [g]}$, $g \in \mathbb N$, the canonical basis of $\mathbb R^g$ and $\mathbb C^g$, respectively. When considering tuples of self-adjoint matrices $(A_1, \ldots, A_g)$ in dimension $d$, we will often write $\mathbf{A}$ to denote this tuple and conversely refer to entries of $\mathbf{A}$ as $A_i$ without specifying this in advance. Moreover, for $\mathbf{s} \in \mathbb R^g$, we write $\mathbf{s}.\mathbf{A}$ to mean the entrywise product $(s_1 A_1, \ldots, s_g A_g)$.

\subsection{Free spectrahedra}
We will now briefly review \emph{free spectrahedra}. These objects are defined by linear matrix inequalities as follows: Let $\mathbf{A} \in (\mathcal{M}_d^{\mathrm{sa}})^{g+1}$, where $g$, $d \in \mathbb N$. We define
\begin{equation*}
    \hat{\mathcal D}_{\mathbf{A}}(n) := \left\{(X_1, \ldots, X_g)\in (\mathcal{M}_n^{\mathrm{sa}})^g~:~ \sum_{i \in [g]} A_i \otimes X_i \leq A_0 \otimes I_{n} \right\}.
\end{equation*}
Then, the free spectrahedron is defined as $\hat{\mathcal D}_{\mathbf{A}} := \bigsqcup_{n \in \mathbb N} \hat{\mathcal D}_{\mathbf A}(n)$. If $A_0 = I_d$, the free spectrahedron is called \emph{monic} and we will write $\mathcal D_{\mathbf{A}}$ instead to distinguish this case. An example of a monic free spectrahedron is the \emph{matrix cube}  $\mathcal D_{\square, g}$ \cite{ben-tal2002tractable}, for which
\begin{equation*}
    \mathcal D_{\square, g}(n) := \left\{(X_1, \ldots, X_g)\in (\mathcal{M}_n^{\mathrm{sa}})^g~:~\|X_i\|_\infty \leq 1~\forall i \in [g]\right\}.
\end{equation*}
This can be written as $\mathcal D_{\mathbf A}$ with $A_i = e_i \oplus -e_i \in \mathcal M_{2g}^{\mathrm{sa}}$ for all $i \in [g]$, where we interpret the vectors as diagonal matrices. The first level of the matrix cube,  i.e.~the set $\mathcal D_{\square, g}(1)$, is just the unit ball of the $\ell_\infty$-norm. 

Another important monic free spectrahedron is the \emph{matrix diamond} $\mathcal D_{\diamond, g}$ \cite{davidson2016dilations}. It's levels are defined as
\begin{equation*}
    \mathcal D_{\diamond, g}(n) := \{(X_1, \ldots, X_g\} \in (\mathcal{M}_n^{\mathrm{sa}})^g ~:~ \sum_{i \in [g]} \epsilon_i X_i \leq I_n ~ \forall \boldsymbol{\epsilon} \in \{\pm 1\}^g \}.
\end{equation*}
Also the matrix diamond can be brought into the form $\mathcal D_{\mathbf A}$ with diagonal matrices $A_i$. At level $1$, the matrix diamond corresponds to the unit ball of the $\ell_1$-norm. Thus, the convex bodies $\mathcal D_{\diamond, g}(1)$ and $\mathcal D_{\square, g}(1)$ are polar duals of each other. 

Both the matrix cube and the matrix diamond are instances of \emph{maximal matrix convex sets}, defined in \cite{davidson2016dilations}. Let $\mathcal C \subset \mathbb R^g$ be a closed convex set. Then, the corresponding matrix convex set is
\begin{align*}
    &\mathcal W_{\mathrm{max}}(\mathcal C)(n):=\\& \left\{(X_1, \ldots, X_g\} \in (\mathcal{M}_n^{\mathrm{sa}})^g ~:~ \sum_{i \in [g]} h_i X_i \leq c I \textrm{~for~all~}(h,c)\mathrm{~s.t.~} \langle h, x \rangle \leq c~\forall x \in \mathcal C\right\}.
\end{align*}
Thus, the tuples $\mathbf{X} \in  \mathcal W_{\mathrm{max}}(\mathcal C)$ have to satisfy all linear inequalities elements of $\mathcal C$ satisfy and $\mathcal W_{\mathrm{max}}(\mathcal C)(1) = \mathcal C$. It can be seen that for any $\mathcal D_{\mathbf A}$ (and more general any matrix convex set) such that $\mathcal D_{\mathbf A}(1) \subseteq \mathcal C$, it holds that $\mathcal D_{\mathbf A} \subseteq \mathcal W_{\mathrm{max}}(\mathcal C)$, which justifies the name \emph{maximal} matrix convex set. If $\mathcal C$ is a polyhedron containing $0$ in its interior, $ \mathcal W_{\mathrm{max}}(\mathcal C)$ is a free spectrahedron since $\mathcal C$ is determined by finally many inequalities $(h^{(j)}, c^{(j)})$, which can be normalized to $c^{(j)}=1$. We can thus speak of a maximal free spectrahedron. It can be brought in the form $\mathcal D_{\mathbf A}$ with diagonal $A_i$ such that the entries of $A_i$ are $h_i^{(j)}$. 

Let $d^\prime \in \mathbb N$ and let $\mathbf{B} \in (\mathcal{M}_{d^\prime}^{\mathrm{sa}})^{g+1}$. We then write
\begin{equation*}
   \hat{ \mathcal D}_{\mathbf A} \subseteq \hat{\mathcal D}_{\mathbf B} \qquad \iff \qquad   \hat{\mathcal D}_{\mathbf A}(n) \subseteq \hat{\mathcal D}_{\mathbf B}(n) \quad \forall n \in \mathbb N.
\end{equation*}
It is known that in general $\hat{\mathcal D}_{\mathbf A}(1) \subseteq \hat{\mathcal D}_{\mathbf B}(1)$ does not imply $\hat{ \mathcal D}_{\mathbf A} \subseteq \mathbf{s}.\hat{\mathcal D}_{\mathbf B}$ with $\mathbf{s}=(1, \ldots, 1)$, but the implication holds if $\mathbf{s}$ is chosen sufficiently small. This motivates the definition of inclusion constant sets.
\begin{defi}
Consider the inclusion constant set
$$\Delta_{\mathbf A}(g,d):=\{ \mathbf{s} \in [0,1]^g\, : \, \mathcal D_{\mathbf A}(1) \subseteq \mathcal D_{\mathbf B}(1) \implies \mathbf{s}.\mathcal D_{\mathbf A} \subseteq \mathcal D_{\mathbf B} \text{ for all } \mathbf{B} \in (\mathcal M_d^{\mathrm{sa}})^g\}$$
and its non-monic version 
$$\hat \Delta_{\mathbf A}(g,d):=\{ {\mathbf s} \in [0,1]^g\, : \, \mathcal D_{\mathbf A}(1) \subseteq \hat{\mathcal D}_{\mathbf B}(1) \implies \mathbf{s}.\mathcal D_{\mathbf A} \subseteq \hat{\mathcal D}_{\mathbf B} \text{ for all } \mathbf{B} \in (\mathcal M_d^{\mathrm{sa}})^{g+1}\}.$$
Here, for a given set of $g$-tuples $\mathcal D$, we write $\mathbf{s}.{\mathcal D}:= \{(s_1 X_1, \ldots, s_g X_g)~:~(X_1, \ldots, X_g) \in \mathcal D\}$.
\end{defi}
Clearly, if $A_0 = I$, we have $\hat \Delta_{\mathbf A}(g,d) \subseteq \Delta_{\mathbf A}(g,d)$ for all parameter choices. We will now show that these sets are the same under mild conditions on $\mathcal D_{\mathbf A}$.

\begin{prop} \label{prop:hat-equals-hatless}
Let $A_1, \ldots, A_g$ be a $g$-tuple of matrices such with the property that 
\begin{equation*}
    0 \in \operatorname{int}\mathcal D_{\mathbf A}(1)
\end{equation*}
Then, for all $d$, $\hat\Delta_{\mathbf A}(g,d) \subseteq \Delta_{\mathbf A}(g,d)$, where $A_0 = I$ in the non-monic case.
\end{prop}
\begin{proof}
Consider $ \mathbf{B} \in (\mathcal M_d^{\mathrm{sa}})^{g+1}$ such that $\mathcal D_{\mathbf A}(1) \subseteq \hat{\mathcal D}_{\mathbf{B}}(1)$. Since $0 \in \mathcal D_{\mathbf A}(1)\subseteq \hat{\mathcal D}_{\mathbf B}(1)$, this implies that $B_0$ is positive semi-definite. Since $0 \in \operatorname{int} \mathcal D_{\mathbf A}(1)$, there is a $\delta > 0$ such that $\pm \delta e_i \in \mathcal D_{\mathbf A}(1)$ for all $i \in [g]$. Thus, for any $i \in [g]$, $\pm \delta e_i \in \mathcal D_{\mathbf A}(1)$ implies that $\pm \delta B_i \leq B_0$ hence $\operatorname{supp} B_i \subseteq \operatorname{supp} B_0$. Setting $\tilde B_i := \left. B_i \right|_{\operatorname{supp} B_0}$, it is clear that $\hat{\mathcal D}_{\mathbf B} = \hat{\mathcal D}_{\mathbf{\tilde B}}$. We have now $\tilde B_0$ which is positive \emph{definite}, and we set, for all $i=0,1,\ldots, g$,
$$C_i:= \tilde B_0^{-1/2}\tilde B_i\tilde B_0^{-1/2}.$$
Clearly, $C_0 = I$, and we have 
$$\hat{\mathcal D}_{\mathbf B} = \hat{\mathcal D}_{\mathbf{\tilde B}}=\hat{\mathcal D}_{\mathbf{C}} = \mathcal D_{\mathbf C},$$
finishing the proof.
\end{proof}
As a shorthand, we will write $\Delta_\diamond(g,d)$ for the inclusion constant set of the matrix diamond $\mathcal D_{\diamond, g}$, $\Delta_\square(g,d)$ for the inclusion constant set of the matrix cube $\mathcal D_{\square, g}$, and $\Delta_{\mathcal C}(d)$ for the inclusion constant set of a maximal free spectrahedron $\mathcal W_{\mathrm{max}}(\mathcal C)(d)$.

For more information on free spectrahedra, we refer the reader to \cite{helton_matricial_2013, davidson2016dilations, helton2019dilations}.

\subsection{Measurement compatibility}
The most general measurements in quantum mechanics are abstractly described by \emph{positive operator-valued measures} (POVMs). Let $d$, $k \in \mathbb N$. Then, a POVM with $k$ outcomes on a $d$-dimensional quantum system is a tuple $(E_1, \ldots, E_k) \subset (\mathcal M_d^+)^k$ such that 
\begin{equation*}
    \sum_{i \in [k]} E_i = I_d.
\end{equation*}
If we would like to perform several measurements on the same quantum state we are looking for another measurement from which the measurement outcomes we are interested in can be obtained via classical post-processing. Since the outcomes of quantum measurements are probabilistic, we allow also the classical post-processing to be probabilistic. Thus, for POVMs $(E_{a|x})_{a \in [k_x]}$, where $k_x \in \mathbb N$ for all $x \in [g]$, $g \in \mathbb N$, we are looking for a finite set $\Lambda$ and a POVM $(E_\lambda)_{\lambda \in \Lambda}$ for which there are conditional probabilities $p(a|x, \lambda)$ such that 
\begin{equation*}
    E_{a|x} = \sum_{\lambda \in \Lambda} p(a|x, \lambda) E_\lambda \qquad \forall a \in [k_x], x \in [g].
\end{equation*}
If such an $(E_\lambda)_{\lambda \in \Lambda}$ exists, the POVMs $(E_{a|x})_{a \in [k_x]}$ are \emph{compatible} or \emph{jointly measurable}.

It is known that not all POVMs are compatible. For example, projective measurements are compatible if and only if the projections commute. In general, the situation is more complex. We refer to \cite{Heinosaari2016} for a review of these questions and for equivalent formulations of measurement compatibility. It is however possible to render measurements compatible by adding noise. For the special case of white noise, we replace $E_{a|x}$ by
\begin{equation*}
    E_{a|x}^\prime := s_x E_{a|x} + (1-s_x) \frac{I_d}{k_x},
\end{equation*}
where $s_x \in [0,1]$. This can be interpreted as a device which performs the desired measurement with probability $s_x$ and outputs a uniformly random number from $[k_x]$ with probability $(1-s_x)$. Fixing the dimension $d$, the number of measurements $g$ and their respective outcomes $\mathbf k = (k_1, \ldots, k_x)$, We can thus ask how much white noise we have to add to make any such collection of measurements compatible. That is the interpretation of the \emph{compatibility region}
\begin{align*}
    &\Gamma(g,\mathbf k, d) := \\&\left\{\mathbf{s} \in [0,1]^g~:~ \mathrm{POVMs~}(E_{a|x}^\prime)_{a \in [k_x]}\mathrm{~are~compatible~}\forall \mathrm{~POVMs~}(E_{a|x})_{a \in [k_x]} \subset (\mathcal M_d^+)^{k_x}\right\}.
\end{align*}
In \cite{bluhm2020compatibility}, we have proven that 
\begin{equation} \label{eq:Gamma-is-jewel}
     \Gamma(g,\mathbf k, d) = \Delta_{\jewel}( g, \mathbf k, d),
\end{equation}
where $\Delta_{\jewel}( g, \mathbf k, d)$ is the inclusion constant set of a maximal free spectrahedron which we called the \emph{matrix jewel} \cite[Definition 4.1]{bluhm2020compatibility} and which generalizes the matrix diamond. In this work, we will mostly consider the case $\mathbf k = 2^{\times g}$, for which $\Delta_{\jewel}(g, 2^{\times g}, d) = \Delta_{\diamond}(g, d)$.

\subsection{Quantum steering}

Given positive integers $g$, $d$, and a vector $\mathbf k:=(k_1, \ldots, k_g) \in \mathbb N^g$, a \emph{$(g,\mathbf k,d)$-assemblage} is a set of positive semidefinite matrices $\boldsymbol{\sigma}=(\sigma_{a|x})_{a \in [k_x], x \in [g]} \subset \mathcal M_d^+$ with the property that
$$\forall x \in [g], \qquad \sum_{a=1}^{k_x} \operatorname{Tr} \sigma_{a|x} = 1.$$
In other words, we ask that for all $x \in [g]$, $\sigma_x:=\sum_{a=1}^{k_x} \sigma_{a|x}$ should be a density matrix. Furthermore, we require the average always to be the same, i.e.\ $\sigma_x = \sigma_y$ for all $x$, $y \in [g]$. We denote by $\mathcal Q(g,\mathbf k,d)$ the set of all such $(g,\mathbf k,d)$-assemblages. It is known that for any $\boldsymbol{\sigma} \in \mathcal Q(g,\mathbf k,d)$ there exists a quantum state $\rho \in \mathcal M_{d^2}^+$ and $d$-dimensional \emph{positive operator-valued measures} (POVMs) $\{E_{a|x}\}_{a \in [k_x]}$ for all $x \in [g]$ such that \cite{Schroedinger1936, Hughston1993, Sainz2015}
\begin{equation}\label{eq:assemblage-POVM}
\sigma_{a|x} = [\operatorname{Tr} \otimes \operatorname{id}]\left( (E_{a|x} \otimes I) \rho \right).
\end{equation}
This can be interpreted as a two party protocol \cite{Wiseman2007, Jones2007}: Alice prepares a bipartite state $\rho$ and sends half of it to Bob. Then, she chooses an $x \in [g]$ and measures $(E_{a|x})_{a \in [k_x]}$ on her part. This prepares the assemblage $\boldsymbol{\sigma}$ on Bob's side.  Conversely, it is clear that any choice of bipartite quantum state $\rho$ and POVMs $(E_{a|x})_{a \in [k_x]}$ yields by \eqref{eq:assemblage-POVM} an assemblage having average $\sigma_x = [\operatorname{Tr} \otimes \operatorname{id}](\rho)$ for all $x \in [g]$. Since $\sigma_x$ is thus independent of $x$, we write $\bar \sigma := \sigma_x$ for the average state.

\bigskip

Let $\mathcal L(g, \mathbf k,d)$ the set of $(g,\mathbf k,d)$-assemblages having a \emph{local hidden state (LHS) model} \cite{Wiseman2007}: there exist a finite set $\Lambda$, probability distributions $p$, $q$ and density matrices $\sigma_\lambda$ such that
$$\forall a \in [k_x],~x \in [g], \qquad \sigma_{a|x} = \sum_{\lambda \in \Lambda} q_\lambda p(a|x,\lambda)\sigma_\lambda.$$
The set $\{q_\lambda, \sigma_\lambda\}_{\lambda \in \Lambda}$ is thus an \emph{ensemble} of quantum states. Note that we need only consider finite $\Lambda$ since we consider fixed $g$ and $\mathbf k$ \cite{Skrzypczyk2014, Cavalcanti2016a}. The operational interpretation of LHS models put forward in \cite{Wiseman2007, Jones2007} is that even if Bob does not trust Alice but only his own measurements, he must conclude that the initial state $\rho$ was entangled if the assemblage $\boldsymbol{\sigma}$ does not admit an LHS model. A case in which the assemblage $\boldsymbol \sigma$ admits an LHS is if the $g$ measurements $(E_{a|x})_{a \in [k_x]}$ Alice applies on her side are compatible \cite{Uola2014,quintino2014joint,uola2015one}. If an assemblage does not allow for an LHS model, we will say that the assemblage is \emph{steerable}.

\bigskip

We recall now the notion of \emph{steering inequalities}. Given a tuple of self-adjoint $d$-dimensional matrices $(F_{a|x})_{a \in [k_x], x \in [g]}$, define the LHS, resp.~quantum, maximal value of $\mathbf{F}$ by
\begin{align*}
V_{\mathcal L}(\mathbf{F}) &:= \sup_{\boldsymbol{\sigma} \in \mathcal L(g, \mathbf k,d)} \sum_{a,x} \operatorname{Tr} (\sigma_{a|x} F_{a|x})\\
V_{\mathcal Q}(\mathbf{F}) &:= \sup_{\boldsymbol{\sigma} \in \mathcal Q(g, \mathbf k,d)} \sum_{a,x} \operatorname{Tr} (\sigma_{a|x} F_{a|x}).
\end{align*}
These tuples $\mathbf{F}$ thus define linear $(g,\mathbf k,d)$-steering inequalities \cite{Cavalcanti2009}. A steering inequality can be used to determine if some $(g,\mathbf k, d)$-assemblages do not admit an LHS model, i.e. are steerable, if $V_{\mathcal L}(\mathbf{F}) < V_{\mathcal Q}(\mathbf{F}) $. This has been used e.g.\ in \cite{Saunders2010} to demonstrate quantum steering experimentally. In this work, all steering inequalities will be linear, so we will only refer to them as steering inequalities.

In addition, we consider $n$-dimensional quantum values, where we restrict the system that Alice can measure to be $n$-dimensional. Let $\mathcal Q^n(g,\mathbf k, d)$ be the set of assemblages $(\sigma_{a|x})_{a \in [k_x], x \in [g]}$ such that there exists a set of $n$-dimensional POVMs $(E_{a|x})_{x \in [k_x]}$, $x \in [g]$ and a state $\rho$ on $\mathbb C^n \otimes \mathbb C^d$ for which
\begin{equation*}
\sigma_{a|x} = [\operatorname{Tr} \otimes \operatorname{id}]\left( (E_{a|x} \otimes I) \rho \right).
\end{equation*}
Note that  $\mathcal Q^d(g,\mathbf k, d) =  \mathcal Q(g,\mathbf k, d)$ by \cite{Schroedinger1936, Hughston1993, Sainz2015}. At the other side of the spectrum, the set $\mathcal Q^1(g,\mathbf k, d)$ coincides with the LHS models for which $\Lambda$ is a singleton. We define
\begin{equation*}
    V^n_{\mathcal Q}(\mathbf{F}) := \sup_{\boldsymbol{\sigma} \in \mathcal Q^n(g, \mathbf k,d)} \sum_{a,x} \operatorname{Tr} (\sigma_{a|x} F_{a|x}).
\end{equation*}

We have the following characterization of the values of a given inequality $\mathbf{F}$ in terms of POVMs. Recall that a POVM $(E_a)_{a \in [k]}$ is called \emph{trivial} if the effect operators $E_a$ are scalar: $E_a = p(a)I$; if this is the case, then $p(\cdot)$ is a probability vector. 

\begin{lem}\label{lem:V-lambda_max}
Let $\mathbf{F}$ be a $(g,\mathbf k,d)$-steering inequality, where $d$, $g \in \mathbb N$, $\mathbf k \in \mathbb N^g$. Then \begin{equation*}
     V^n_{\mathcal Q}(\mathbf{F}) = \quad\quad \sup_{\mathbf{E} \text{ POVMs in dim.\ }n} \quad \lambda_{\max} \left[ \sum_{a,x}  
E_{a|x} \otimes F_{a|x} \right],
\end{equation*}
where the POVMs are given by effect elements $E_{a|x} \in \mathcal M_n^+$, with 
$$\forall x \in [g], \qquad \sum_{a=1}^{k_x} E_{a|x} = I_d.$$
In particular,
\begin{alignat*}{2}
V_{\mathcal L}(\mathbf{F}) &= \quad\sup_{\mathbf{E} \text{ trivial POVMs}}\quad &&\lambda_{\max} \left[ \sum_{a,x}  
E_{a|x} \otimes F_{a|x} \right]\\
 V_{\mathcal Q}(\mathbf{F}) &= \quad\quad \sup_{\mathbf{E} \text{ POVMs in dim.\ }d} \quad&& \lambda_{\max} \left[ \sum_{a,x}  
E_{a|x} \otimes F_{a|x} \right].
\end{alignat*}
\end{lem}

\begin{proof}
Let us start with the second formula, corresponding to the $n$-dimensional quantum value, $n \geq 2$:
\begin{align*}
V^n_{\mathcal Q}(\mathbf{F}) &= \sup_{\boldsymbol{\sigma} \in \mathcal Q^n(g, \mathbf k,d)} \sum_{a,x} \operatorname{Tr} (\sigma_{a|x} F_{a|x}) \\
&= \sup_{\rho, \mathbf{E}} \operatorname{Tr}\left[\sum_{a,x} [\operatorname{Tr} \otimes \operatorname{id}]\left( (E_{a|x} \otimes I) \rho \right) F_{a|x} \right]\\
&= \sup_{\rho, \mathbf{E}} \operatorname{Tr}\left[\sum_{a,x} (E_{a|x} \otimes I) \rho  (I \otimes F_{a|x}) \right]\\
&= \sup_{\rho, \mathbf{E}} \operatorname{Tr}\left[\sum_{a,x} (E_{a|x} \otimes F_{a|x}) \rho \right]\\
&= \sup_{ \mathbf{E}} \lambda_{\max}\left[\sum_{a,x} E_{a|x} \otimes F_{a|x} \right].
\end{align*}
Here, $\mathbf{E}$ is a collection of $n$-dimensional POVMs and we have used \eqref{eq:assemblage-POVM} in the second equality. For the LHS case, we have 
	\begin{align}
		V_{\mathcal L}(\mathbf{F}) &= \sup_{\boldsymbol{\sigma} \in \mathcal L(g, \mathbf k,d)} \sum_{a,x} \operatorname{Tr} (\sigma_{a|x} F_{a|x}) \nonumber \\
		&= \sup_{\mathbf{p},\mathbf{q},\boldsymbol{\sigma}} \sum_{\lambda,a,x} q_\lambda p(a|x, \lambda) \operatorname{Tr}(F_{a|x}\sigma_\lambda) \nonumber\\
		&=\sup_{\mathbf{p},\boldsymbol{\sigma}} \sum_{a,x}  p(a|x) \operatorname{Tr}(F_{a|x}\sigma) \nonumber \\
		&=\sup_{\mathbf{p}} \lambda_{\max} \left[\sum_{a,x}  p(a|x) F_{a|x}\right] \label{eq:looks-like-previous-work}, 
	\end{align}
	where we have used convexity in $q_\lambda$ for the third equality and the fact that the extreme points of the probability simplex are Dirac masses. It is easy to see that this expression for $V_{\mathcal L}(\mathbf{F})$ coincides with $V^{1}_{\mathcal Q}(\mathbf{F})$. The conclusion follows by identifying the conditional probability $p(a|x)$ with a set of trivial POVMs.
\end{proof}

Expressions similar to \eqref{eq:looks-like-previous-work} for steering inequalities involving Pauli matrices on qubit systems appear in \cite{Saunders2010, Evans2014}. Note that the choice of $\mathbf{F}$ fixes the tuple $(g, \mathbf k,d)$. We are interested in the \emph{largest quantum violation} of the steering inequality $\mathbf{F}$: $V_{\mathcal Q}(\mathbf{F})/V_{\mathcal L}(\mathbf{F})$. The quantity appears in \cite{Hsieh2016} under the name of steering fraction. We shall also consider the steering inequality $\mathbf{F}$ giving the largest violation, for a fixed setting described by $g$, $\mathbf k = (k_1, \ldots, k_g)$, and $d$
$$\gamma_{g,\mathbf k, d} = \sup_{\mathbf{F}} \frac{V_{\mathcal Q}(\mathbf{F})}{V_{\mathcal L}(\mathbf{F})}.$$
 If $\mathbf k = 2^{\times g}$, we will drop the index $\mathbf k$ and write $\gamma_{g,d}$ instead. Restricting the $\mathbf{F}$ to unbiased steering inequalities, i.e. such that  $F_{-|x} = - F_{+|x}$ for all $x \in [g]$, we obtain the largest unbiased quantum violation $\gamma_{g,d}^0$. Here, we have written for the outcomes $a \in \{\pm\}$ instead of $a \in [2]$, which becomes a helpful mnemonic in the next section. One of the goals of this work is to compute the quantities $\gamma_{g,d}$ and $\gamma^0_{g,d}$.

\section{Spectrahedral formulation of quantum steering} \label{sec:spec-form}

The aim of this section is to make the connection between quantum steering and the inclusion of (free) spectrahedra.
In this work, we shall only consider the dichotomic case, $\mathbf{k} = 2^{\times g}$, $g \in \mathbb N$. Recall the definition of the matrix cube
$$\mathcal D_{\square,g}:=\{(X_1, \ldots, X_g) \in (\mathcal M_n^{\mathrm{sa}})^g \, : \, \|X_i\|_\infty \leq 1 \, \forall i \in [g]\}.$$

Consider a dichotomic steering inequality $\mathbf{F} = (F_{\pm|x})_{x \in [g]}$ to which we associate the matrices $\{A_{\pm|x}\}_{x \in [g]}$ defined by
\begin{equation}\label{eq:F-A}
F_{+|x} = A_{+|x} + A_{-|x} \qquad \text{ and } \qquad  F_{-|x} = A_{+|x} - A_{-|x} \qquad \forall x \in [g]
\end{equation}
and the non-monic free spectrahedron
\begin{equation}\label{eq:def-hat-D-F}
\hat{\mathcal D}_{\mathbf{\tilde F}}(n):=\left\{(X_1, \ldots, X_g) \in (\mathcal M_n^{\mathrm{sa}})^g \, : \, \sum_{x=1}^g A_{-|x} \otimes X_x \leq (I_d - \sum_{x \in [g]} A_{+|x}) \otimes I_n \right\}.
\end{equation}
Note that the unbiased case corresponds to having $A_{+|x} = 0$ for all $x \in [g]$.

\begin{thm} \label{thm:steering-equals-inclusion}
	For an arbitrary dichotomic steering inequality $\mathbf{F}$, we have	
	\begin{equation*}
	    \mathcal D_{\square,g}(n) \subseteq \hat{\mathcal D}_{\mathbf{\tilde F}}(n) \iff V^n_{\mathcal Q}(\mathbf{F}) \leq 1 \qquad \forall n \in [d].
	\end{equation*}
    In particular,
	\begin{enumerate}
		\item $\mathcal D_{\square,g}(1) \subseteq \hat{\mathcal D}_{\mathbf{\tilde F}}(1) \iff V_{\mathcal L}(\mathbf{F}) \leq 1$.
		\item $\mathcal D_{\square,g} \subseteq \hat{\mathcal D}_{\mathbf{\tilde F}} \iff V_{\mathcal Q}(\mathbf{F}) \leq 1$.
	\end{enumerate}
\end{thm}
\begin{proof}
	Let us first consider the LHS value and the level $1$ inclusion. 
From Lemma \ref{lem:V-lambda_max}, we have 
	$$V_{\mathcal L}(\mathbf{F})  = \sup_{\mathbf{p}} \lambda_{\max} \left[\sum_{a,x}  p(a|x) F_{a|x}\right].$$
	In the dichotomic case we are considering, for a given $x$, we have 
	$$\sum_{a\in \{\pm\}}  p(a|x) F_{a|x} = (2p(+|x)-1)A_{-|x} + A_{+|x}.$$
	The extremal points of the set of conditional probabilities correspond to 
	$2p(+|x)-1 = \epsilon_x$, and thus, we have 
	$$\sup_{\mathbf{p}} \lambda_{\max} \left[\sum_{a,x}  p(a|x) F_{a|x}\right] = \sup_{\boldsymbol{\epsilon} \in \{\pm 1\}^g} \lambda_{\max} \left[ \sum_{x=1}^g \epsilon_x A_{-|x} + \sum_{x=1}^g A_{+|x}\right].$$
	Hence, 
	$$V_{\mathcal L}(\mathbf{F}) \leq 1 \iff \mathcal D_{\square, g}(1) \subseteq \hat{\mathcal D}_{\mathbf{\tilde F}}(1).$$
	
	We now turn to the $n$-dimensional quantum value of the dichotomic steering inequality and to the free spectrahedral inclusion. Again, starting from Lemma \ref{lem:V-lambda_max}, we have
	\begin{align*}
		V^n_{\mathcal Q}(\mathbf{F}) 
		&=\sup_{\mathbf{E}} \lambda_{\max}\left[ \sum_{x=1}^g E_{+|x} \otimes F_{+|x} + E_{-|x} \otimes F_{-|x}\right]\\
		&=\sup_{0 \leq E_{+|x} \leq I_n} \lambda_{\max}\left[ \sum_{x=1}^g (2E_{+|x}-I_n) \otimes A_{-|x} + I_n \otimes\sum_{x=1}^g A_{+|x} \right]\\
		&=\sup_{-I_n \leq X_x \leq I_n} \lambda_{\max}\left[ \sum_{x=1}^g X_x \otimes A_{-|x} + I_n \otimes\sum_{x=1}^g A_{+|x}\right],
	\end{align*}
	and thus $V^n_{\mathcal Q}(\mathbf{F}) \leq 1 \iff \mathcal D_{\square, g}(n) \subseteq \hat{\mathcal D}_{\mathbf{\tilde F}}(n)$, finishing the proof. 
\end{proof}

\begin{remark}
Note that in the above proof for $V_{\mathcal Q}(\mathbf{F})$, we make use of the fact that $\mathcal D_{\square,g} \subseteq \hat{\mathcal D}_{\mathbf{\tilde F}} \iff \mathcal D_{\square,g}(d) \subseteq \hat{\mathcal D}_{\mathbf{\tilde F}}(d)$. This follows from the fact that $\mathcal D_{\mathbf{A}} \subseteq \mathcal D_{ \mathbf B}$ is equivalent to $\mathcal D_{\mathbf A}(d) \subseteq \mathcal D_{\mathbf B}(d)$ for $d$-dimensional $\mathbf B$ \cite{helton2019dilations, bluhm2018joint}. Another way to see the former is to note that without loss of generality the measurement preparing the assemblage on Alice's side can be chosen to be $d$-dimensional \cite{Schroedinger1936, Hughston1993, Sainz2015}. The above theorem also shows that the intermediate levels of the spectrahedral inclusion correspond to intermediate degrees of steerability.
\end{remark}

Theorem \ref{thm:steering-equals-inclusion} implies in particular that we can determine whether $V_{\mathcal Q}(\mathbf{F}) \leq 1$ by means of a semidefinite program, since the inclusion of free spectrahedra can be checked in this way \cite{helton_matricial_2013, helton2019dilations, davidson2016dilations}. This agrees with the fact that it can be verified using an SDP that a fixed assemblage admits a LHS model \cite{Skrzypczyk2014, Cavalcanti2016a}. The SDPs discussed in these works admit dual formulations in terms of violations of certain steering inequalities. Note that a difference between that line of work and ours is that the former fix the assemblage and look for the maximal violation for a certain family of steering inequalities, whereas we fix the steering inequality and consider the maximal violation for assemblages in fixed dimensions and with a fixed number of dichotomic measurements on Alice's side. Even more important is the distinction at the level of the maximal steering inequality violation and steering robustness, discussed in Sections \ref{sec:incl-csts} and \ref{sec:relative-size}. The constant sets $\Delta_{\square}(g,d)$, as well as the values $\gamma_{g,d}$ and $\gamma^0_{g,d}$ establish fundamental limits of quantum theory, and they cannot be related to semidefinite programs.

\section{Connecting steering and inclusion constants}\label{sec:incl-csts}

In this section, we extend the connection between quantum steering and the spectrahedral inclusion found in the previous section and establish that the set of inclusion constants for the matrix cube has a connection to the ratio of the largest violation of a steering inequality.

To this end, we introduce the steering constants, which can be understood as a measure of robustness for the capacity of steering inequalities to detect steerability: $V_{\mathcal L}(\mathbf{F}) < V_{\mathcal Q}(\mathbf{F})$. First, notice that there is a class of steering inequalities which we call \emph{trivial}:
$$\forall x\in [g], \, \forall a,b\in[k_x], \qquad F_{a|x} = F_{b|x}=:F_x.$$
For such trivial inequalities, we clearly have 
$$V_{\mathcal L}(\mathbf{F}) = V_{\mathcal Q}(\mathbf{F}) = \lambda_{\max}\left[ \sum_{x=1}^g F_x \right].$$
To an arbitrary steering inequality $\mathbf{F}$, we associate its trivial realization $\mathbf{F^{(0)}}$
$$F^{(0)}_{a|x} := \frac{1}{k_x} \sum_{a=1}^{k_x} F_{a|x}.$$
We consider convex mixtures between steering inequalities and their trivial realizations: for $\mathbf{s} \in [0,1]^g$ we set
$$\mathbf{F^{(\mathbf{s})}} := \mathbf{s}.\mathbf{F} + (\mathbf{1-s}).\mathbf{F^{(0)}},$$
where $\mathbf{1-s}:= (1-s_1, \ldots, 1-s_g)$.
Note that in the dichotomic case, using the notations $A_{\pm|x} = (F_{+|x} \pm F_{-|x})/2$ from \eqref{eq:F-A}, we have 
$$A^{(\mathbf{s})}_{+|x} = A_{+|x} \quad \text{ and } \quad A^{(\mathbf{s})}_{-|x} = s_x A_{-|x} \qquad \forall x \in [g].$$
Moreover, in the dichotomic \emph{unbiased} case, i.e. $F_{-|x} = - F_{+|x}$ for all $x \in [g]$, the trivial realizations are null and $A_{+|x} = F_{+|x}$, hence the convex mixtures above correspond to scaling: $\mathbf{F^{(\mathbf{s})}} = \mathbf{s}.\mathbf{F}$. This fact will allow us to compare our work to previous results in \cite{marciniak2015unbounded}. For the rest of the section, we focus again on dichotomic steering inequalities.

\begin{defi}
	The set of \emph{steering constants} is defined as:
\begin{equation*}
    \Sigma(g,d) := \{\mathbf{s} \in [0,1]^g \, : \, \forall (F_{\pm|1}, \ldots, F_{\pm|g}) \in (\mathcal M_d^{\mathrm{sa}})^{2g}, \, V_{\mathcal L}(\mathbf{F}) \leq 1 \implies  V_{\mathcal Q}(\mathbf{F^{(\mathbf{s})}}) \leq 1\}
\end{equation*}
Let $\Sigma_0$ be set of \emph{unbiased steering constants}, i.e.\ the above implication only has to hold for $\mathbf{F}$ such that $F_{+|x} = - F_{-|x}$.
\end{defi}

We prove now the main theorem of this section.

\begin{thm}\label{thm:Sigma-equals-Delta}
	For all $g$,  $d \in \mathbb N$, $\Sigma_0(g,d) = \Sigma(g,d) = \Delta_\square(g,d) =  \hat \Delta_{\square}(g,d)$.
\end{thm}
\begin{proof}
Note that since $0 \in \operatorname{int}\mathcal D_{\square,g}(1)$, it follows from Proposition \ref{prop:hat-equals-hatless} that $\hat{\Delta}_\square(g,d) = \Delta_\square(g,d)$. Thus, we need only prove $\Sigma(g,d) = \hat{\Delta}_\square(g,d)$ and $\Sigma_0(g,d) = \Delta_\square(g,d)$. We start by showing the former. 
First, we claim that $\mathbf{s}.\mathcal D_{\square,g} \subseteq \hat{\mathcal D}_{\mathbf{F}}$ is equivalent to $\mathcal D_{\square,g} \subseteq \hat{\mathcal D}_{\mathbf{F^{(\mathbf{s})}}}$: for some $ \mathbf{X} \in \mathcal D_{\square, g}(n)$, $n \in \mathbb N$,
\begin{align*}
\mathbf{s}.\mathbf{X} \in \hat{\mathcal D}_{\mathbf F} &\iff \sum_{x=1}^g A_{-|x} \otimes s_xX_x \leq (I - \sum_{x=1}^g A_{+|x})\otimes I_n\\
&\iff \sum_{x=1}^g A^{(\mathbf{s})}_{-|x} \otimes X_x \leq (I - \sum_{x=1}^g A^{(\mathbf{s})}_{+|x})\otimes I_n\\
&\iff \mathbf{X} \in \hat{\mathcal D}_{\mathbf{F^{(\mathbf{s})}}}.
\end{align*}

Using Theorem \ref{thm:steering-equals-inclusion} for the vertical equivalences, the conclusion follows easily: for $\mathbf{s} \in [0,1]^g$,
\begin{center}
\begin{tabular}{rcccl}
$\mathbf{s} \in \Sigma(g,d) \iff \big[\quad \forall \mathbf{F}\quad$    &  $ V_{\mathcal L}(\mathbf{F}) \leq 1$ & $ \implies$ & $V_{\mathcal Q}(\mathbf{F^{(\mathbf{s})}}) \leq 1$&$ \big]$\\
&$\Updownarrow$&&$\Updownarrow$&\\
$\mathbf{s} \in  \hat{\Delta}_\square(g,d) \iff \big[\quad \forall \mathbf{F}\quad $    &  $\mathcal D_{\square,g}(1) \subseteq \hat{\mathcal D}_{\mathbf{F}}(1)$ & $ \implies$ & $\mathcal D_{\square,g} \subseteq \hat{\mathcal D}_{\mathbf{F^{(\mathbf{s})}}}$&$ \big]$\\
\end{tabular}
\end{center}
The second assertion follows using almost the same argument.
\end{proof}
Since we have shown that the set of steering constants and the set of inclusion constants for the matrix cube are equal, this has given us a powerful tool to compute bounds on the former, since the matrix cube is a well-studied object in convex optimization. We will discuss the implications for quantum steering in Section \ref{sec:discussion}. For the remainder of this section, let us focus on the unbiased case.
We are interested in the quotient $V_{\mathcal Q}(\mathbf{F})/V_{\mathcal L}(\mathbf{F})$. The following lemma shows that this quotient is always finite if we set $0/0 = 1$.

\begin{lem}\label{lem:LHS-zero}
Let $\mathbf{F}$ be a dichotomic and unbiased steering inequality. Then, $V_{\mathcal L}(\mathbf{F}) \geq 0$ with equality if and only if $\mathbf{F} \equiv 0$. In particular, $V_{\mathcal L}(\mathbf{F}) = 0$ if and only if $V_{\mathcal Q}(\mathbf{F}) = 0$.
\end{lem}
\begin{proof}
If $\mathbf{F} \equiv 0$, everything is clear, so let us assume that this is not the case. Recall from the proof of Theorem \ref{thm:steering-equals-inclusion} that 
$$V_{\mathcal L}(\mathbf{F})  = \sup_{\boldsymbol{\epsilon} \in \{\pm 1\}^g} \lambda_{\max} \Bigg[ \underbrace{\sum_{x=1}^g \epsilon_x F_{+|x}}_{=:F_{\boldsymbol \epsilon}}\Bigg].$$
We claim that at least one of the $F_{\boldsymbol \epsilon}$ must be non zero; otherwise, for all $x \in [g],$
$$F_{+|x} = \frac{1}{2^{g-1}} \sum_{\boldsymbol{\epsilon} \, : \, \epsilon_x = 1} F_{\boldsymbol \epsilon} = 0.$$
Consider thus $F_{\boldsymbol \epsilon} \neq 0$. Then, either $\lambda_{\max}[F_{\boldsymbol \epsilon}] > 0$ (proving the claim), or $\lambda_{\min}[F_{\boldsymbol \epsilon}] < 0$, in which case we conclude by observing that
$$\lambda_{\max}[F_{-\boldsymbol{\epsilon}}] = \lambda_{\max}[-F_{\boldsymbol{\epsilon}}] = -\lambda_{\min}[F_{\boldsymbol{\epsilon}}] > 0.$$
\end{proof}
We are interested in the set of steering constants for unbiased steering inequalities, since the largest point on the diagonal in $\Sigma_{0}(g,d)$ is related to the maximal violation of steering inequalities considered in \cite{marciniak2015unbounded}, as we show now. 

\begin{prop} \label{prop:sigma-is-delta}
	For any $s \in [0,1]$, $s(1, \ldots, 1) \in \Sigma_0(g,d)$ if and only if for all unbiased $(g, 2^{\times g}, d)$-steering inequalities $\mathbf{F}$ it holds that 
	$$V(\mathbf{F}) := \frac{V_{\mathcal Q}(\mathbf{F})}{V_{\mathcal L}(\mathbf{F})} \leq \frac 1 s,$$
	where we define $V(\mathbf{F}) = 1$ if $V_{\mathcal{L}}(\mathbf{F}) = V_{\mathcal{Q}}(\mathbf{F}) = 0$. In particular, the largest such $s$ is equal to the largest unbiased quantum violation $\gamma^0_{g, d}$.
\end{prop}
\begin{proof}
	Let us prove first the ``$\implies$'' direction. Consider an $s$ as in the statement and a steering inequality $\mathbf{F}$. If $s = 0$ or $V_{\mathcal L}(\mathbf{F}) = 0$, the assertion holds by Lemma \ref{lem:LHS-zero}. Therefore, let us assume, that neither is the case. Let $v:=V_{\mathcal L}(\mathbf{F}) > 0$; we have $V_{\mathcal L}(v^{-1}\mathbf{F}) = 1$. By the definition of the set $\Sigma_0(g,d)$ and since $\mathbf{F}^{(s(1, \ldots, 1))}=s\mathbf{F}$ for unbiased steering inequalities, we then have $V_{\mathcal Q}(s/v\mathbf{F}) \leq 1$. Putting everything together, we obtain
	$$V(\mathbf{F}) = \frac{V_{\mathcal Q}(\mathbf{F})}{V_{\mathcal L}(\mathbf{F})}  \leq \frac{v/s}{v} = \frac{1}{s},$$
	as claimed. The reverse implication is proven in a similar manner.
\end{proof}

\section{Connection to steering robustness} \label{sec:relative-size}

In this section, we interpret the inclusion constants for the matrix cube and thus the steering constants as a robustness measure for quantum steering.  To do this, let us fix an arbitrary steering scenario, given by a Hilbert space dimension $d \in \mathbb N$, and an outcome vector $\mathbf k = (k_1, \ldots, k_g) \in \mathbb N^g$, $g \in \mathbb N$. Importantly, we also fix a density matrix $\bar \sigma \in \mathcal M_d^{\mathrm{sa}}$, which corresponds to the average state $\bar \sigma = \sum_a \sigma_{a|x}$. We consider the sets of all assemblages and of the assemblages admitting an LHS model with \emph{fixed} average state $\bar \sigma$, i.e.\ 
\begin{align*}
\mathcal L_{\bar \sigma}(g,\mathbf k, d) &:= \left\{(\sigma_{a|x})_{a \in [k_x], x \in [g]}\, : \, \sigma_{a|x} = \sum_{\lambda \in \Lambda} q_\lambda p(a|x,\lambda)\sigma_\lambda \, , \, \sum_{\lambda \in \Lambda} q_\lambda \sigma_\lambda = \bar \sigma  \right\},\\
    \mathcal Q_{\bar \sigma}(g,\mathbf k, d) &:= \left\{(\sigma_{a|x})_{a \in [k_x], x \in [g]} \, : \, \forall x \in [g], \, \sum_{a \in [k_x]} \sigma_{a|x} = \bar \sigma  \right\}.
\end{align*}
Here, $\Lambda$ is a finite index set, $\{q_\lambda, \sigma_\lambda\}_{\lambda \in \Lambda}$ an ensemble of quantum states, $p(a|x,\lambda)$ are conditional probabilities and $\sigma_{a|x} \in \mathcal M_d^+$ for all $a \in [k_x]$, $x \in [g]$. Clearly, we have $\operatorname{supp} \sigma_{a|x} \subseteq \operatorname{supp} \bar \sigma$ for all $a \in [k_x]$, $x \in [g]$. Hence, in the case where $\bar \sigma$ is not positive definite, we can restrict all the elements of the assemblage $\boldsymbol \sigma$ to its support, reducing the effective dimension; in what follows, we shall thus assume that $\bar \sigma$ is a positive definite matrix.

We have that $\mathcal L_{\bar \sigma}(g,\mathbf k, d) \subseteq \mathcal Q_{\bar \sigma}(g,\mathbf k, d)$, and our goal is to understand the relative size of these two sets, which can be understood as quantifying the amount of steering possible in the scenario given by $\mathbf k$ and Hilbert space dimension $d$. Abusing notation, we consider the following reference point, which is to be understood as a kind of ``center'' of these sets: 
\begin{equation} \label{eq:center}
    \boldsymbol{\bar \sigma} := \left(\sigma_{a|x} = \frac{\bar \sigma}{k_x}\right)_{a \in [k_x], x \in [g]} \in \mathcal L_{\bar \sigma}(g,\mathbf k, d) \subseteq \mathcal Q_{\bar \sigma}(g,\mathbf k, d).
\end{equation}

For a fixed positive definite matrix $\bar \sigma$, the following linear maps establish an isomorphism between the set of all assemblages with average $\bar \sigma$  and the set of $g$-tuples of POVMs having $k_1, \ldots, k_g$ outcomes, respectively. 
\begin{align}
    \{\text{assemblages with average $\bar \sigma$}\} &\longleftrightarrow \{\text{$g$-tuples of POVMs}\} \nonumber\\
    (\sigma_{a|x})_{a,x} &\longmapsto E_{a|x} := \bar \sigma^{-1/2} \sigma_{a|x} \bar \sigma^{-1/2} \label{eq:isomorphism}\\
    \bar \sigma^{1/2} E_{a|x} \bar \sigma^{1/2} =:\sigma_{a|x} &\longmapsfrom (E_{a|x})_{a,x}.\nonumber
\end{align}
This isomorphism allows us to relate steerability of states and the compatibility of measurements as in \cite{uola2015one} and leads to an interpretation of the sets of inclusion constants. To achieve this, we consider the robustness of steerability with respect to a certain form of noise. There have been several suggestions to quantify the amount of steerability in the literature such as e.g\ \cite{Skrzypczyk2014, Piani2015, Cavalcanti2016, Hsieh2016, Bavaresco2017, jencova2018incompatible, Ku2018}. See also \cite{Cavalcanti2016a} for a review. We will consider a noise model in which we take convex mixtures with $\boldsymbol{\bar \sigma}$ as in \eqref{eq:center}, which is called reduced-state steering robustness in \cite{Cavalcanti2016} (c.f.\ also \cite{Jencova2017}). It can be interpreted as adding white noise $I_d/k_x$ to the measurements $E_{a|x}$ on Alice's side when she prepares the assemblage $(\sigma_{a|x})_{a \in [k_x], x \in [g]}$. 

\begin{prop}
An assemblage with given average $\bar \sigma$, $\bar \sigma > 0$,  has an LHS model if and only if the corresponding $g$-tuple of POVMs is compatible. Taking convex combinations with the ``trivial'' assemblage corresponds to mixing in white noise for POVMs. We have, for all invertible density matrices $\bar \sigma \in \mathcal M_d^+$,
$$\{\mathbf{s} \in  [0,1]^g \, : \, \mathbf{s}.\mathcal Q_{\bar \sigma}(g,\mathbf k, d) + (\mathbf{1-s}).\boldsymbol{\bar \sigma} \subseteq \mathcal L_{\bar \sigma}(g,\mathbf k, d) \} = \Delta_{\jewel}(g,\mathbf k, d).$$
\end{prop}
\begin{proof}
Let us prove the first of the three claims. Consider an assemblage $\boldsymbol \sigma$ having an LHS model for some finite index set $\Lambda$: 
$$\sigma_{a|x} = \sum_{\lambda \in \Lambda} q_\lambda p(a|x,\lambda) \sigma_\lambda.$$
On the POVM side, this corresponds to 
$$E_{a|x} = \bar \sigma^{-1/2} \sigma_{a|x} \bar \sigma^{-1/2} = \sum_{\lambda \in \Lambda}  p(a|x,\lambda) \underbrace{q_\lambda \bar \sigma^{-1/2} \sigma_\lambda  \bar \sigma^{-1/2}}_{=: E_{\lambda}}.$$
It is easy to see that the $(E_\lambda)_\lambda$ form a joint POVM for the $(E_{a|x})_{a,x}$: the only property that needs checking is normalization: $$\sum_{\lambda \in \Lambda} E_\lambda = \sum_{\lambda \in \Lambda} q_\lambda \bar \sigma^{-1/2} \sigma_\lambda  \bar \sigma^{-1/2} = \bar \sigma^{-1/2} \left(\sum_{\lambda \in \Lambda} q_\lambda  \sigma_\lambda \right) \bar \sigma^{-1/2} = I.$$

Conversely, let $(E_\lambda)_{\lambda \in \Lambda}$ be a joint POVM for the $g$-tuple of POVMs $E_{a|x}$; we shall assume that $E_\lambda \neq 0$ for all $\lambda$. There exists then conditional probabilities $p(a|x,\lambda)$ such that 
$$\bar \sigma^{-1/2} \sigma_{a|x} \bar \sigma^{-1/2}  = \sum_{\lambda \in \Lambda} p(a|x,\lambda) E_\lambda.$$
It follows that, for all $a \in [k_x]$, $x \in [g]$,
$$\sigma_{a|x}  =  \sum_\lambda \underbrace{\Tr(\bar \sigma E_\lambda)}_{=:q_\lambda} p(a|x,\lambda) \underbrace{\bar \sigma^{1/2} \frac{E_\lambda}{\Tr(\bar \sigma E_\lambda)} \bar \sigma^{1/2}}_{=:\sigma_\lambda},$$
which provides an LHS model for the assemblage $\sigma_{a|x}$. 

For the remaining claims, it remains to see that the isomorphism \eqref{eq:isomorphism} maps 
\begin{equation} \label{eq:noise-isomorphism}
    s_x E_{a|x} + (1-s_x) \frac{I}{k_x} \longleftrightarrow  s_x \sigma_{a|x} + (1-s_x) \frac{\bar \sigma}{k_x} 
\end{equation}
for any $\mathbf{s} \in [0,1]^g$. By \eqref{eq:Gamma-is-jewel}, $\mathbf{s} \in  \Delta_\diamond(g, \mathbf k, d)$ implies that all corresponding noisy POVMs are compatible. The isomorphism together with \eqref{eq:noise-isomorphism} and the first claim thus imply that for any assemblage $(\sigma_{a|x})_{a \in [k_x], x \in [g]}$ with average state $\bar \sigma$, the noisy assemblage $\mathbf{s}.(\sigma_{a|x})_{a,x} + (\mathbf{1-s}).\boldsymbol{\bar \sigma}$ with the same average state admits an LHS model. Thus 
\begin{equation*}
    \mathbf{s}.\mathcal Q_{\bar \sigma}(g,\mathbf k, d) + (\mathbf{1-s}).\boldsymbol{\bar \sigma} \subseteq \mathcal L_{\bar \sigma}(g,\mathbf k, d)
\end{equation*}
The reverse inclusion follows in a similar manner.

\end{proof}

So far, we have only found an interpretation of $\Delta_\diamond(g,d)$ in terms of steering robustness, while we are looking for an interpretation of $\Delta_\square(g,d)$. The next proposition shows that the two sets are actually equal. Note that the proposition does not follow from the duality of matrix convex sets as considered in \cite{davidson2016dilations}, but is a stronger statement, since we fix the dimension $d$.
\begin{prop} \label{prop:Deltas-are-equal}
Let $\mathcal C \subset \mathbb R^g$ be a polytope containing $0$ in its interior, $g \in \mathbb N$. Then, for all $d \in \mathbb N$
\begin{equation*}
    \Delta_{\mathcal C}(d) = \Delta_{\mathcal C^\circ}(d),
\end{equation*}
where $\mathcal C^\circ$ is the polar of $\mathcal C$. In particular, for all $g$, $d \in \mathbb N$, it holds that 
\begin{equation*}
    \Delta_\diamond(g,d) = \Delta_\square(g,d).
\end{equation*}
\end{prop}
\begin{proof}
Let $\mathcal R \subset \mathbb R^{g}$ be a finite set such that for $x \in \mathbb R^g$
\begin{equation} \label{eq:set-of-hyperplanes}
    \langle h, x \rangle \leq 1 \quad\forall h \in \mathcal R \iff x \in \mathcal C, 
\end{equation}
which is equivalent to $\mathcal C = \mathcal R^\circ$.
In other words, $\mathcal R$ is a finite collection of hyperplanes determining $\mathcal C$. Note that since $0 \in \operatorname{int} \mathcal C$, such finite collection of hyperplanes exist and moreover, $\mathcal C^\circ = \operatorname{conv} \mathcal R$ (see \cite[Section IV.1]{Barvinok2002} and the bipolar theorem \cite[Theorem IV.1.2]{Barvinok2002}). 
We know that $\mathcal W_{\mathrm{max}}(\mathcal C)(1) \subset \mathcal D_{\mathbf{B}}(1)$ if and only if $\sum_{i \in [g]} x_i B_i \leq I$ for all $\mathbf{x}\in \mathcal C$, $B_i \in \mathcal M_d^{\mathrm{sa}}$. This is equivalent to $ \mathbf{B} \in \mathcal W_{\mathrm{max}}(\mathcal C^\circ)(d)$ since the extreme points of $\mathcal C$ are the hyperplanes defining $\mathcal C^\circ$ as in \eqref{eq:set-of-hyperplanes} \cite[Section IV.1]{Barvinok2002}. Thus, $\mathbf{s} \in    \Delta_{\mathcal C}(d)$ if and only if 
\begin{equation*}
    \sum_{i = 1}^g  B_i \otimes s_i X_i \leq I
\end{equation*}
for all $\mathbf{B} \in \mathcal W_{\mathrm{max}}(\mathcal C^\circ)(d)$, $\mathbf{X} \in \mathcal W_{\mathrm{max}}(\mathcal C)(d)$. Note that we only need to check inclusion at level $d$ by \cite{helton2019dilations, bluhm2018joint}.
In the same way, since $\mathcal C^{\circ \circ} = \mathcal C$ by the bipolar theorem \cite[Theorem IV.1.2]{Barvinok2002},  $\mathbf{s} \in    \Delta_{\mathcal C^\circ}(d)$ if and only if
\begin{equation*}
    \sum_{i = 1}^g  B_i \otimes s_i X_i \leq I
\end{equation*}
for all $\mathbf{X} \in \mathcal W_{\mathrm{max}}(\mathcal C^\circ)(d)$, $\mathbf{B} \in \mathcal W_{\mathrm{max}}(\mathcal C)(d)$. Since the order of the tensor factors does not matter, the assertion follows.
\end{proof}

Thus, we have found the desired interpretation of the inclusion constants for the matrix cube in terms of robustness of quantum steering, see Figure \ref{fig:assemblage-sets}:
\begin{cor}
Let $g$, $d \in \mathbb N$. For or all invertible density matrices $\bar \sigma \in \mathcal M_d^+$,
$$\{\mathbf{s} \in  [0,1]^g \, : \, \mathbf{s}.\mathcal Q_{\bar \sigma}(g, 2^{\times g}, d) + (\mathbf{1-s}).\boldsymbol{\tilde \sigma} \subseteq \mathcal L_{\bar \sigma}(g, 2^{\times g}, d) \} = \Delta_{\square}(g, d).$$
\end{cor}

\begin{figure}
    \centering
    \includegraphics{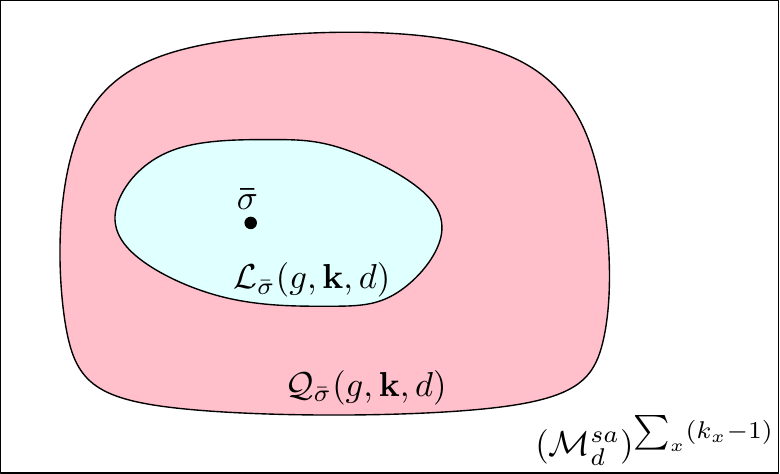}
    \caption{Inclusion of the set of assemblages having a LHS model into the set of all quantum assemblages. The inclusion set $\Delta_\square$ measures the amount by which one has to shrink $\mathcal Q_{\bar \sigma}$ (around $\bar \sigma$) in order to make it fit inside $\mathcal L_{\bar \sigma}$.}
    \label{fig:assemblage-sets}
\end{figure}

Together with the findings in Section \ref{sec:incl-csts}, we have thus found two interpretations of the set of inclusions constants for the matrix cube, one as steering robustness and the other as a violation of steering inequalities. It is reasonable that these two things are connected, since robustness measures which can be formulated as SDPs often have a dual formulation which can be interpreted as the maximal violation of certain steering inequalities \cite{Skrzypczyk2014, Cavalcanti2016, Cavalcanti2016a}. Again, we would like to emphasize that although the robustness of a given assemblage or the maximal value of a given steering inequality can be formulated as SDPs, the sets $\Delta_\square(g,d)$, and the corresponding constants $\gamma_{g,d}$ and $\gamma^0_{g,d}$ cannot: they correspond to worse-case scenarios (for assemblages or steering inequalities), and those optimization problems cannot be readily formulated as SDPs.

\section{Inclusions constants for the complex matrix cube} \label{sec:cube-csts}

In Section \ref{sec:incl-csts}, we have connected the set of steering constants $\Sigma(g,d)$ to the set of inclusion constants for the complex matrix cube,  $\Delta_\square(g,d)$. For the real matrix cube, this set was studied in \cite{ben-tal2002tractable, helton2019dilations}. We follow the strategy outlined in these works to obtain bounds on $\Delta_\square(g,d)$ for the complex matrix cube.

\subsection{\texorpdfstring{$g$}{g}-independent bounds}\label{sec:g-indep-LB}
Let us start by exhibiting lower bounds for the set $\Delta_\square(g,d)$ which only depend on the Hilbert space dimension $d$ and which are independent of the cardinality $g$ of dichotomic assemblages. 
	
First, following the proofs in  \cite{ben-tal2002tractable}, we obtain, for fixed $d$ and all $g \geq 1$,
$$\tau_*(d)(\underbrace{1, 1, \ldots, 1}_{g \text{ times}}) \in \Delta_\square(g,d),$$
where the constant $\tau_*(d)$ is related to the complex version of the quantity $ \vartheta(d)$ considered in \cite{helton2019dilations} (see also \cite[Proposition 8.3 and Remark 8.4]{bluhm2020compatibility}). It is defined by
\begin{equation}\label{eq:def-s-star-d}
\tau_*(d)=\frac{1}{\vartheta_{\mathbb C}(d)}:=\min_{\stackrel{B \in \mathcal M_d^{\mathrm{sa}}}{\|B\|_1 = 1}} \operatorname{\mathbb E} |\langle z, B z \rangle| = \min_{\stackrel{b \in \mathbb R^d}{\|b\|_1 = 1}} \operatorname{\mathbb E}\left| \sum_{i=1}^d b_i |z_i|^2\right|.
\end{equation}
In other words, $\tau_*(d)$ is the largest constant such that
$$\tau_*(d) \|B\|_1 \leq \operatorname{\mathbb E} |\langle z, B z \rangle|, \qquad \forall B \in \mathcal M_d^{\mathrm{sa}}.$$
The expectation above is considered with respect to standard complex Gaussian vectors $z \in \mathbb C^d$. Note that a similar expression could have been written by using random vectors uniformly distributed on the unit sphere of $\mathbb C^d$, at the price of normalizing the $1$-Schatten (resp.~the $\ell_1$) norm by $d$; this is this approach followed in \cite{helton2019dilations}.

\begin{thm}[\cite{ben-tal2002tractable}]
For any $g$, $d \in \mathbb N$, $d \geq 2$, and any $g$-tuple $\mathbf{B} \in (\mathcal M_d^{\mathrm{sa}})^g$, we have
$$\mathcal D_{\square,g}(1) \subseteq \mathcal D_{\mathbf{B}}(1) \implies \tau_*(d) \mathcal D_{\square,g} \subseteq \mathcal D_{\mathbf{B}}.$$
\end{thm}

\begin{proof}
Consider a $g$-tuple $B$ as in the statement, satisfying
$$\sum_{i=1}^g \epsilon_i B_i \leq I_d, \qquad \forall \boldsymbol{\epsilon} \in \{\pm 1\}^g.$$
The matrix cube free spectrahedron can be defined as $\mathcal D_{\mathbf A}$, where $A_i$ is a diagonal $2g \times 2g$ matrix, having $\pm 1$ in the positions $2i-1$, resp.~$2i$. The inclusion $ s \mathcal D_{\square,g} \subseteq \mathcal D_{\mathbf{B}}$ is equivalent, by \cite{helton2019dilations} (see also \cite[Lemma 4.4]{bluhm2018joint}), to the existence of a unital completely positive map $\Phi: \mathcal M_{2g} \to \mathcal M_d$ such that $\Phi(A_i) = s B_i$, for all $i \in [g]$. In turn, we can formulate this problem as a semidefinite program (SDP): the inclusion $s\mathcal D_{\square,g} \subseteq \mathcal D_{\mathbf{B}}$ holds if and only if
\begin{align}
\label{eq:SDP-primal} 1/s \geq \qquad \qquad \min \quad &t\\
\nonumber    \text{subj. to}\quad &\operatorname{Tr}_{2g} C = t I_d\\
\nonumber         &\operatorname{Tr}_{2g}(C \cdot A_i^\top \otimes I) = B_i, \quad \forall i \in [g]\\
\nonumber    & C \text{ is positive semidefinite}.
\end{align}

Here, $C \in \mathcal M_{2g} \otimes \mathcal M_d$ and the partial traces are over the first factor.
The matrix $C$ plays the role of the Choi matrix of the map $s^{-1}\Phi$ above. Let us compute the dual of the SDP above. Consider the Lagrangian 
\begin{align*}
    &L(t,C,\Lambda_0, \Lambda_1, \ldots, \Lambda_g, M)\\ :=& t + \langle \Lambda_0, \Tr_{2g}C - t I_d\rangle + \sum_{i=1}^g \langle \Lambda_i, B_i - \Tr_{2g}(C\cdot A_i^\top \otimes I_d)\rangle - \langle C, M \rangle\\
    =& \sum_{i=1}^g \langle \Lambda_i, B_i \rangle  + t(1-\Tr \Lambda_0) + \langle C, I \otimes \Lambda_0 - M - \sum_{i=1}^g \bar A_i \otimes \Lambda_i \rangle,
\end{align*}
where $M, \Lambda_i$ are Lagrange multipliers, with $M \geq 0$. 
The dual SDP reads
\begin{align*}
 \max \quad & \sum_{i=1}^g \langle \Lambda_i, B_i \rangle\\
     \text{subj. to}\quad &\operatorname{Tr} \Lambda_0 = 1\\
     &\sum_{i=1}^g \bar A_i \otimes \Lambda_i \leq I \otimes \Lambda_0.
\end{align*}  
Note that strong duality holds since Slater's condition is verified for the dual: $\Lambda_0 = I/d$, $\Lambda_1 = \cdots = \Lambda_g = 0$ is a strict feasible point for the dual SDP. The inequality constraints in the dual SDP can be restated as $\pm \Lambda_i \leq \Lambda_0$ for all $i \in [g]$, hence $\rho:=\Lambda_0 \geq 0$.
Using a similar argument as in \cite[Lemma 2.2]{ben-tal2002tractable}, we optimize separately over each $\Lambda_i$:
\begin{align*}&\max \{  \langle \Lambda_i , B_i\rangle \, : \,  -\rho \leq \Lambda_i \leq \rho\} \\=& \max \{  \langle \Lambda_i ,B_i\rangle \, : \,  -I \leq \rho^{-1/2}\Lambda_i\rho^{-1/2} \leq I\}\\
=& \max \{  \langle \rho^{-1/2}\Lambda_i\rho^{-1/2} , \rho^{1/2}B_i\rho^{1/2}\rangle \, : \,  -I \leq \rho^{-1/2}\Lambda_i\rho^{-1/2} \leq I\}\\
=& \max \{  \langle \rho^{-1/2}\Lambda_i\rho^{-1/2} , \rho^{1/2}B_i\rho^{1/2}\rangle \, : \,  \|\rho^{-1/2}\Lambda_i\rho^{-1/2}\|_\infty \leq 1\}\\
=& \|\rho^{1/2}B_i\rho^{1/2}\|_1.
\end{align*} 
Note that we could assume $\rho>0$ by perturbing it slightly and making the perturbation arbitrarily small in the end. Hence, the original problem reduces to 
\begin{align*}
 1/s \geq \qquad \qquad  \max \quad & \sum_{i=1}^g \|\rho^{1/2} B_i \rho^{1/2}\|_1\\
     \text{subj. to}\quad & \rho \geq 0, \,  \operatorname{Tr} \rho = 1.
\end{align*}  
To this end, fix a density matrix $\rho$ and a complex Gaussian vector $z$. With 
$$\epsilon_i := \operatorname{sign} \langle z, \rho^{1/2} B_i \rho^{1/2} z \rangle \in \{\pm 1\},$$
we have
$$\sum_{i=1}^g \epsilon_i B_i \leq I \implies \sum_{i=1}^g \epsilon_i  \langle z, \rho^{1/2} B_i \rho^{1/2} z \rangle \leq  \langle z, \rho z \rangle.$$
Taking the expectation over normally distributed random $z$, we get
$$\mathbb E \sum_{i=1}^g  |\langle z, \rho^{1/2} B_i \rho^{1/2} z \rangle| \leq \mathbb E \langle z, \rho z \rangle = \operatorname{Tr} \rho = 1.$$
Using $\mathbb E|\langle z, X z \rangle| \geq \tau_*(d) \|X\|_1$ for self-adjoint matrices $X \in \mathcal M_d^{\mathrm{sa}}$, we can conclude: $s=\tau_*(d)$ satisfies the inequality \eqref{eq:SDP-primal} for all choices of $B$. 
\end{proof}

Let $K_d \subseteq \mathbb R^d$ be the \emph{cross polytope}, the unit ball of the $\ell_1^d$ norm. The optimization problem in \eqref{eq:def-s-star-d} corresponds to minimizing the convex function 
\begin{align*}
\mathcal H:	\partial K_d &\to \mathbb R\\
	b &\mapsto \operatorname{\mathbb E} \left| \sum_{i=1}^d b_i |z_i|^2\right|.
\end{align*}
Note however that the domain of the function above is the unit sphere of the $\ell_1^d$ norm, hence it is \emph{not} convex, making the minimization problem non-trivial. We are going to make use of the convexity of $\mathcal H$ on the faces of the polytope $K_d$. 

The extreme points of $K_d$ are the $2d$ points $\pm e_i$, with $i \in [d]$. Here, the $e_i$ are the canonical basis vectors. The facets of $K_d$ are $(d-1)$-simplices indexed by the $2^d$ subsets $A \subseteq [d]$:
$$H_A := \operatorname{conv} \left( \{e_i\}_{i \in A} \cup \{-e_j\}_{j \in [d] \setminus A} \right).$$
Above, the subset $A$ marks the set of indices corresponding to positive coordinates. The symmetric group acts naturally on the faces $H_A$:
$$\sigma . b = \left(\epsilon^A_1 |b_{\sigma(1)}|, \epsilon^A_2 |b_{\sigma(2)}|, \ldots, \epsilon^A_d |b_{\sigma(d)}| \right),$$
where 
$$\epsilon^A_i = \begin{cases}
1 &\quad \text{ if } i \in A\\
-1 &\quad \text{ if } i \notin A.
\end{cases}$$
The orbit of a point $b \in H_A$ is given by $d!$ (not necessarily distinct) points having as an average the ``central'' point of $H_A$, $d^{-1} \epsilon^A$. Using the convexity of $\mathcal H$ on $H_A$, we have, for any $b \in H_A$, 
$$\mathcal H(d^{-1}\epsilon^A) = \mathcal H\left(\frac{1}{d!} \sum_{\sigma \in \mathfrak S_d} \sigma.b \right) \leq \frac{1}{d!}\sum_{\sigma \in \mathfrak S_d} \mathcal H(\sigma.b) = \mathcal H(b),$$
where we have used the fact that $\mathcal H(b) = \mathcal H(\sigma.b)$, which follows from the permutation invariance of the standard complex Gaussian distribution. Hence, the minimum of $\mathcal H$ on the compact set $\partial K_d$ is attained at one of the points $d^{-1}\epsilon^A$. Using again the permutation invariance, we can restrict the search to the subsets of the form 
$$\epsilon^A = (\underbrace{1, 1, \ldots, 1}_{k \text{ times}}, \underbrace{-1, -1, \ldots, -1}_{d-k \text{ times}}),$$
for some integer $0 \leq k \leq d$. We show in the next lemma\footnote{We thank Franck Barthe for allowing us to reproduce his elegant proof of the result in the even $d$ case. A proof of the general case, using a similar transform, was given on MathOverflow by user ``Fedor Petrov'' \cite{MO1}} that the minimum is attained at the most ``balanced'' values of $k$. 

\begin{lem}\label{lem:minimum-signed-chi2}
Consider $d$ of i.i.d.~complex Gaussian random variables $z_1, \ldots, z_d \in \mathbb C$. The minimum of the function
\begin{align*}
    \{0,1, \ldots, d\} &\to \mathbb R_+\\
    k &\mapsto \mathbb E \left| \sum_{i=1}^k |z_i|^2 - \sum_{i=k+1}^d |z_i|^2 \right|
\end{align*}
is attained at $k = \lceil d/2 \rceil$. 
\end{lem}
\begin{proof}
We shall treat separately the cases where $d$ is odd or even. In the case where $d$ is even, the result follows from the following claim, which is more general:

\noindent\textbf{Claim.} Given $2n$ integrable i.i.d.~random variables $Z_i$, we have, for all $0 \leq k \leq 2n$, 
$$\mathbb E | \underbrace{Z_1 + \cdots + Z_k - Z_{k+1} - \cdots - Z_{2n}}_{V_k} | \geq \mathbb E| \underbrace{Z_1 + \cdots + Z_n - Z_{n+1} - \cdots - Z_{2n}}_{V_n} |.$$

To prove the claim, let us consider the $\phi_Z(u) := \mathbb E \exp(\mathrm{i}uZ)$, the \emph{characteristic function} of one of the random variables  $Z_i$. Since, for all real $t$, we have (see \cite[Eq. (21.16)]{bronstein})
$$|t| = \frac{2}{\pi}\int_0^\infty \frac{1-\cos(tu)}{u^2} \, \mathrm{d}u,$$
it follows that, for an integrable random variable $V$, we have
\begin{equation}\label{eq:E-abs-Re-phi}
\mathbb E|V| = \frac{2}{\pi}\int_0^\infty \frac{1-\mathbb E\cos(uV)}{u^2} \, \mathrm{d}u = \frac{2}{\pi} \int_0^\infty \frac{1-\Re\phi_V(u)}{u^2} \, \mathrm{d}u.    
\end{equation}
In our setting, we get 
$$\phi_{V_k}(u) = \phi_Z(u)^k\phi_Z(-u)^{2n-k} = \phi_Z(u)^k\overline{\phi_Z(u)}^{2n-k},$$
hence
$$\Re \phi_{V_k}(u) \leq |\phi_{V_k}(u)| = |\phi_{Z}(u)|^{2n} = \phi_{V_n}(u) = \Re \phi_{V_n}(u).$$
Plugging the above inequality in \eqref{eq:E-abs-Re-phi} yields the claim, and thus the result for even $d$. 

\medskip

We now move to the case of odd $d$, which is more intricate\footnote{An elegant proof of the general case of non-negative random variables was given on MathOverflow by user ``Steve'' \cite{MO2}}. We shall make use of \cite[Proposition 12.8]{helton2019dilations}, which states that, for all integers $s,t$ with $s \geq t$,
\begin{equation}\label{eq:fst-increasing}
f_{s,t}(p_{s,t}) \leq f_{s+2,t-2}(p_{s+2,t-2}),
\end{equation}
where the function $f_{s,t}$ is defined by
$$f_{s,t}(p) = \operatorname{\mathbb E}\left| a\sum_{i=1}^s x_i^2 - b\sum_{j=1}^{t} x_j^2 \right|,$$
where $D:=s+t$, and $a,b \geq 0$ are determined by $sa+tb = D$ and $p = b/(a+b)$. Here, the vector $(x_1, \ldots, x_D)$ has uniform distribution on the unit sphere sphere of $\mathbb C^D$; one could have considered a standard complex Gaussian vector, with an extra factor of $D$. The value $p_{s,t}$ is the point $p$ where the function $f_{s,t}(p)$ attains its minimum on $[0,1]$; it is shown in \cite[Proposition 9.2]{helton2019dilations} that for even $D$, $p_{s,t} = 1/2$, hence $a=b=1$.
Going back to our setting, with $d=2n+1$ odd, \eqref{eq:fst-increasing} translates to the fact that the function 
\begin{align*}
    \{n+1, n+2, \ldots, 2n+1\} &\to \mathbb R_+\\
    k &\mapsto \operatorname{\mathbb E}\left| \sum_{i=1}^k |z_i|^2 - \sum_{j=k+1}^{2n+1} |z_j|^2 \right|\\
    &\quad = \operatorname{\mathbb E}\left| \sum_{i=1}^{2k} x_i^2 - \sum_{j=2k+1}^{4n+2} x_j^2 \right|
\end{align*}
is increasing. Indeed, we have used $s:=2k \geq t:= 2(2n+1-k)$, which holds as soon as $k \geq n+1 = \lceil d/2 \rceil$. In this case, $D = 4n+2$ is even. The proof is complete.
\end{proof}

We now have all the ingredients to state the main result of this section. 

\begin{thm} \label{thm:g-independent-bounds}
	The parameter $\tau_*(d)$, giving the $g$-independent inclusion constant for the complex matrix cube, is given by
	\begin{equation}
		\tau_*(d) = 
		4^{-n}\binom{2n}{n}, \qquad \text{ with } n:= \lfloor d/2 \rfloor.
	\end{equation}
\end{thm}
\begin{proof}
	Let us first consider the case of even $d=2n$. Using Lemma \ref{lem:minimum-signed-chi2}, we have to evaluate 
	$$\frac{1}{2n}\operatorname{\mathbb E}\left| \sum_{i=1}^n |z_i|^2 - \sum_{j=n+1}^{2n} |z_j|^2 \right| = \frac{1}{4n}\operatorname{\mathbb E}\left| \sum_{i=1}^n |x_i|^2+|y_i|^2 - \left(\sum_{j=n+1}^{2n} |x_j|^2+|y_j|^2  \right) \right|,$$
	where we have written $z_j = (x_j + \mathrm{i} y_j)/{\sqrt 2}$, where $x_1, \ldots, x_{2n}, y_1, \ldots, y_{2n}$ are i.i.d.~standard \emph{real} Gaussian random variables. In other words, we have to compute $\operatorname{\mathbb E}|A-B|$, where $A,B$ are i.i.d.~$\chi^2$ random variables with $2n$ degrees of freedom. 
	
	The distribution of $A-B$ has been considered in \cite{dist-diff-chi-squared}, we reproduce the main ideas below. The moment generating function of either of the random variables $A,B$ is given by \cite[Section 2.6.1]{ross2019introduction}
	$$M_A(t) = \operatorname{\mathbb E} \exp(tA) = (1-2t)^{-n}.$$
	Hence, the moment generating function of the difference $A-B$ is 
	$$M_{A-B}(t) = M_A(t)M_B(-t) = (1-4t^2)^{-n} = \left( \frac{1/4}{1/4 - t^2} \right)^n,$$
	where one recognizes the formulas for the \emph{variance-gamma distribution}, see \cite[Eq.~(24)]{haas2009financial}  or \cite[Section 4.1.1]{kotz2012laplace}
	$$M_{VG}(t) = \exp(\mu t) \left( \frac{\alpha^2-\beta^2}{\alpha^2 - (\beta+t)^2} \right)^\lambda,$$
	with parameters $\mu = 0$, $\lambda = n$, $\alpha = 1/2$, $\beta = 0$. The variance-gamma distribution has probability density function on the real line given by
	$$f_{A-B}(x) = \frac{1}{\sqrt \pi 4^n \Gamma(n)} |x|^{n-1/2} K_{n-1/2}(|x|/2),$$
	where $K$ is the modified Bessel function of the second kind. Using \cite[10.43.E19]{DLMF}, we compute 
	$$\int_{-\infty}^\infty |x| f_{A-B}(x) \,  \mathrm{d}x = \frac{4}{\sqrt \pi} \frac{\Gamma(n+1/2)}{\Gamma(n)},$$
	which leads in turn to 
	$$\tau_*(2n) = \frac{\Gamma(n+1/2)}{\sqrt \pi \Gamma(n+1)},$$
	the announced formula for the $d=2n$ even case (see also \cite[Theorem 13.1(b)]{helton2019dilations}).
	
	\medskip

	Let us now consider the $d=2n+1$ odd case. From Lemma \ref{lem:minimum-signed-chi2}, we have
	$$\tau_*(2n+1) = \frac{1}{2n+1}\operatorname{\mathbb E}\left| \sum_{i=1}^{n+1} |z_i|^2 - \sum_{j=n+2}^{2n+1} |z_j|^2 \right| = \frac{1}{2(2n+1)}\operatorname{\mathbb E}\left| \sum_{i=1}^{2n+2} x_i^2- \sum_{j=2n+3}^{4n+2} x_j^2\right|.$$
	The formula in the statement easily follows from the following claim: 
	
\noindent\textbf{Claim.} Given $4n+4$ i.i.d.~standard Gaussian random variables $x_i$, $y_j$, we have
$$\operatorname{\mathbb E}\left| \sum_{i=1}^{2n+2} x_i^2- \sum_{j=1}^{2n} y_j^2\right| = \operatorname{\mathbb E}\left| \sum_{i=1}^{2n+2} x_i^2- \sum_{j=1}^{2n+2} y_j^2\right|$$

To prove the claim,	let us consider the general case of the difference of two independent $\chi^2$ random variables, with $2s$, resp.~$2t$ degrees of freedom:
$$D_{s,t}:=\sum_{i=1}^{2s} x_i^2 - \sum_{j=1}^{2t} y_j^2,$$
where $x_i,y_j$ are i.i.d.~standard Gaussians. The probability distribution function $\mathrm{d}D_{s,t}/\mathrm{d}x$ of $D_{s,t}$ can be explicitly computed with the help of Whittaker special functions \cite[Eq.~(5)]{klar2015note} (using that $W[\kappa, \mu, x] = W[\kappa, -\mu, x]$):
\begin{align*}
    \frac{\mathrm{d}D_{s,t}}{\mathrm{d}x} = \begin{cases}
    \displaystyle\frac{x^{(s + t)/2 - 1}}{\Gamma(s)2^{s + t}} W[(s - t)/2, (s + t - 1)/2, x]&\qquad \text{ if } x \geq 0\\
    \displaystyle\frac{|x|^{(s + t)/2 - 1}}{\Gamma(t)2^{s + t}} W[(t - s)/2, (s + t - 1)/2, |x|]&\qquad \text{ if } x < 0.
    \end{cases}
\end{align*}
One can then compute $\mathbb E |D_{s,t}|$ using \cite[13.23.4]{DLMF} and \cite[15.8.1]{DLMF}:
\begin{equation*}
    2 \binom{s+t}{s-1}\pFq{2}{1}{t,s + t + 1,}{t + 2}{-1}+2 \binom{s+t}{t-1}\pFq{2}{1}{s,s + t + 1,}{s + 2}{-1}
\end{equation*}
where $\pFq{2}{1}{a,b}{c}{z}$ is the hypergeometric function. The fact that $\mathbb E |D_{s,s}| = \mathbb E |D_{s,s-1}|$ for $s\geq 2$ can be checked by direct computation. Indeed, using Kummer's theorem \cite[15.4.26]{DLMF}, one has 
\begin{equation} \label{eq:kummer}
   \mathbb E |D_{s,s}| = 4\binom{2s}{s+1} \pFq{2}{1}{2s + 1,s}{s + 2}{-1} = 4^{1-s}s\binom{2s}{s}. 
\end{equation}
For $\mathbb E |D_{s,s-1}|$, use first  \cite[15.5.16]{DLMF} to compute
$$\pFq{2}{1}{2s,s-1}{s + 1}{-1} = \frac 1 2 \pFq{2}{1}{2s - 1,s-1}{s + 1}{-1} + \frac{1}{s+1}\pFq{2}{1}{2s,s-1}{s + 2}{-1}$$
and then  \cite[15.5.14]{DLMF} to obtain
$$\pFq{2}{1}{2s,s}{s + 2}{-1} = \frac{2s}{s-1} \pFq{2}{1}{2s + 1,s}{s + 2}{-1} - \frac{s+1}{s-1}\pFq{2}{1}{2s,s}{s + 1}{-1}.$$
Finally, express the functions on the right hand side using again Kummer's theorem to conclude. The assertion follows by \eqref{eq:kummer} with $s=n+1$.
\end{proof}

\subsection{Optimality of the bound}

In this section, we show that the constant $\tau_*(d)$ from \eqref{eq:def-s-star-d} is optimal. We start with a lemma.
\begin{lem} \label{lem:nets}
Let $d \in \mathbb N$ and let $B \in \mathcal M_d^{\mathrm{sa}}$. Let $K\in \mathbb N$, $\delta > 0$ and let $\{U_i\}_{i \in [K]}$ be an $\delta$-net on $(\mathcal U(d), \norm{\cdot}_\infty)$. Then, there exists a probability distribution $(u_1, \ldots, u_K)$ such that
\begin{equation*}
 \forall v \in \mathbb C^d, \, \|v\|_2 =1, \qquad    \left| \int_{\mathcal U(d)} |\langle v, U^\ast B U v \rangle| d\mu(U) - \sum_{i = 1}^K u_i |\langle v, U_i^\ast B U_i v \rangle|  \right| \leq  2 \delta \norm{B}_1,
\end{equation*}
where $\mu$ be the normalized Haar measure on $\mathcal U(d)$.
\end{lem}
\begin{proof}
Let $\{U_i\}_{i \in [K]}$ be as in the assertion. Then, $\{\mathcal B_\delta(U_i)\}_{i \in [K]}$ covers $\mathcal U(d)$, where 
\begin{equation*}
    \mathcal B_\delta(U_i) := \{V\in \mathcal U(d): \norm{V - U_i}_\infty \leq \delta\}.
\end{equation*}
From this cover, we can construct a disjoint cover of cardinality $ K$ by removing intersections
\begin{equation*}
    \mathcal S_1 := B_\delta(U_1), \qquad \mathcal S_i := \mathcal B_\delta(U_i) \setminus \left(\bigcup_{j = 1}^{i-1} \mathcal S_j\right).
\end{equation*}
Set $u_i := \mu(\mathcal S_i)$. Note that some of the $\mathcal S_i$ might be empty, such that $b_i = 0$. It is clear that $u_i \geq 0$ since the Haar measure is positive and 
\begin{equation*}
    \sum_{i = 1}^J u_i = \mu\left(\bigcup_{i = 1}^K \mathcal S_i\right) = \mu(\mathcal U(d)) = 1
\end{equation*}
by normalization. Thus,
\begin{align*}& \left| \int_{\mathcal U(d)} |\langle v, U^\ast B U v \rangle| d\mu(U) - \sum_{i = 1}^K u_i |\langle v, U_i^\ast B U_i v \rangle|  \right| \\
    &\leq\sum_{i = 1}^K \int_{\mathcal S_i} \Big| |\langle v, U^\ast B U v \rangle| - |\langle v, U_i^\ast B U_i v \rangle| \Big| d\mu(U)\\
    &\leq\sum_{i = 1}^K \int_{\mathcal S_i} |\langle v, (U^\ast B U - U_i^\ast B U_i) v \rangle| d\mu(U)\\
    &\leq\sum_{i = 1}^K \int_{\mathcal S_i} \norm{U^\ast B U - U_i^\ast B U_i}_\infty d\mu(U) \\
    &\leq 2 \delta \norm{B}_1 \sum_{i = 1}^J \int_{\mathcal S_i} d\mu(U) = 2 \delta \norm{B}_1
\end{align*}
Here, we have used the triangle inequality for the first inequality, the reverse triangle inequality for the second one, and the fact that $\mathcal S_i \subseteq \mathcal B_\delta(U_i)$ in the last inequality. 
\end{proof}
\begin{remark}
It follows from  \cite[Lemma 9.5]{ledoux1991probability} that there exist $\delta$-nets as above with cardinality
\begin{equation*}
    K \leq \left(1 + \frac{2}{\delta}\right)^{2d}.
\end{equation*}
\end{remark}

We show now, following the idea from \cite[Section 5.4]{helton2019dilations}, that the constant $\tau_*(d)$ from \eqref{eq:def-s-star-d} is optimal. To this end, let us fix a dimension $d$ and a constant $s > \tau_*(d)$. Hence, there exists a matrix $B$ such that 
$$s \|B\|_1 > \mathbb E_z |\langle z, Bz \rangle| = d \mathbb E_\phi |\langle \phi, B\phi \rangle|,$$
where the first expectation is over a standard complex Gaussian vector $z \in \mathbb C^d$, while the second one is over a uniform vector $\phi$ on the unit sphere of $\mathbb C^d$. Moreover, we shall normalize $B$ such that $\mathbb E_\phi |\langle \phi, B\phi \rangle| = 1$.

\begin{prop} \label{prop:optimality-of-constant}
For a matrix $B$ as above, there exists, for a $\delta$ small enough, a $\delta$-net $(U_1, \ldots, U_K)$ on $(\mathcal U(d), \norm{\cdot}_\infty)$, and a probability distribution $(u_1, \ldots, u_K)$  such that for $A_i := (1+2 \delta \norm{B}_1)^{-1} u_i U_i^*BU_i$, we have 
$$\mathcal D_{\square, g}(1) \subseteq \mathcal D_{\mathbf A}(1)\qquad \text{and} \qquad s\mathcal D_{\square, g} \nsubseteq \mathcal D_{\mathbf A}.$$
\end{prop}
\begin{proof}
Let us prove the first inclusion. We shall use $\{U_1, \ldots, U_K\}$ and $\{u_1, \ldots, u_K\}$ as in Lemma \ref{lem:nets} for a (small) value of $\delta$ which we shall fix later. Let $\boldsymbol{\epsilon} \in \{\pm 1\}^g = \mathcal D_{\square,g}(1)$. For any  unit vector $v \in \mathbb C^d$, we have 
\begin{align*}
(1 + 2\delta \norm{B}_1)\langle v, \sum_{i=1}^g \epsilon_i A_i v \rangle &= \langle v, \sum_{i=1}^g u_i \epsilon_i U_i^*BU_i v \rangle\\
&\leq\sum_{i=1}^g u_i |\langle v,   U_i^*BU_i v \rangle| \\
&\leq \int |\langle v,   U^*BU v \rangle| \mathrm{d}\mu(U) + 2\delta \norm{B}_1\\
&\leq \int |\langle \phi,   B \phi \rangle| \mathrm{d}\phi + 2\delta \norm{B}_1 \\
&= 1+ 2\delta \norm{B}_1,
\end{align*}
so $\lambda_{\max}(\sum_i \epsilon_i A_i) \leq 1$. Here, we have used Lemma \ref{lem:nets}.

Consider now the $g$-tuple $\mathbf{X} \in \mathcal D_{\square,g}(d)$ given by the unitary operators 
\begin{equation*}
    X_i = (U_i^* \operatorname{Pol}(B)^*U_i)^\top.
\end{equation*}
Here, $\operatorname{Pol}(B)$ is the unitary arising from the polar decomposition as $B = \operatorname{Pol}(B)|B|$ and we take the transposition with respect to the canonical basis. We have, using the maximally entangled state $\mathbb C^d \otimes \mathbb C^d \ni \Omega = d^{-1/2} \sum_{i=1}^d e_i \otimes e_i$:
$$(1 + 2\delta \norm{B}_1) \langle \Omega, \sum_{i=1}^g A_i \otimes s X_i \Omega \rangle = \frac s d \sum_{i=1}^g \Tr[U_i^*BU_i\cdot X_i^\top] =  \frac s d \|B\|_1 > 1,$$
proving that $s\mathbf{X}$ is not an element of $\mathcal D_{\mathbf A}(d)$ if we choose $\delta$ small enough. Indeed, $\frac s d \|B\|_1 = 1 + \nu$ for some $\nu > 0$, so it suffices to choose $\delta < \nu/(2\norm{B}_1)$.
\end{proof}

\subsection{\texorpdfstring{$d$}{d}-independent bounds}

The previous section provided inclusions constants independent of the dimension of the matrix cube $g$, but dependent on the dimension $d$. In this section, we will consider the converse situation, where the inclusion constants depend on $g$, but not on $d$. Theorem 6.7 of \cite{passer2018minimal} asserts that for any $g$, $d \in \mathbb N$,
\begin{equation*}
    \Delta_\square(g,d) \supseteq \left\{\mathbf{s} \in  [0,1]^g~:~\sum_{i = 1}^g s_i^2 \leq 1\right\}.
\end{equation*}
In fact, the converse inclusion holds for $d \geq 2^{\lceil (g-1)/2\rceil }$, as we will show now. The next proposition, which appears in \cite{marciniak2015unbounded} and is proven here for convenience, shows that these lower bounds are tight. Let $(F_{\pm|x})_{x\in [g]}$ be a family of anti-commuting self-adjoint unitary operators such that $\{F_{+|x},F_{+|x'}\} := F_{+|x} F_{+|x^\prime} + F_{+|x^\prime} F_{+|x} = 2 \delta_{x,x^\prime} I$ and choose $F_{-|x} = - F_{+|x}$ for all $x \in [g]$. Such operators can be constructed recursively as $F_1^{(0)} = [1]$ and 
\begin{equation*}
    F_i^{(k+1)} = \sigma_X \otimes F_i^{(k)} \quad \forall i \in [2k+1], \qquad F_{2k+2}^{(k+1)} =  I_{2^k} \otimes \sigma_Y,\qquad F_{2k+3}^{(k+1)} =  I_{2^k} \otimes \sigma_Z.
\end{equation*}
Here, $\sigma_X$, $\sigma_Y$ and $\sigma_Z$ are the Pauli matrices. The dimension dependence of this construction is essentially optimal \cite{newman1932note, hrubes2016families}. See also \cite{passer2018minimal} and Section 8.2 of \cite{bluhm2018joint}. The construction generalizes a steering inequality for qubits that already appears in \cite{Cavalcanti2009, Pusey2013}.
\begin{prop} \label{prop:d-independent-bounds}
Let  $d \geq 2^{\lceil (g-1)/2\rceil }$, $g \geq 2$. Then, 
\begin{equation*}
    \Sigma(g,d) = \Delta_\square(g,d) = \left\{\mathbf{s} \in  [0,1]^g~:~\sum_{i = 1}^g s_i^2 \leq 1\right\}.
\end{equation*}
Generally, for any $d \in \mathbb N$, $\Sigma(g,d) = \Delta_\square(g,d) \supseteq \left\{\mathbf{s} \in  [0,1]^g~:~\sum_{i = 1}^g s_i^2 \leq 1\right\}$.
\end{prop}
\begin{proof}
We only need to prove that $\Delta_\square(g,d) \subseteq \left\{\mathbf{s} \in  [0,1]^g~:~\sum_{i = 1}^g s_i^2 \leq 1\right\}$, since the reverse inclusion follows from \cite[Theorem 6.7]{passer2018minimal} and Theorem \ref{thm:Sigma-equals-Delta}. Let $(F_{\pm|x})_{x \in [g]}$ be as above and let $\mathbf{s} \in  [0,1]^g$. Then, using Lemma \ref{lem:V-lambda_max},
\begin{align*}
   V_{\mathcal L}(\mathbf{s}.\mathbf{F}) &= \sup_{\mathbf{p}} \lambda_{\max}\left[ \sum_{x \in [g]} s_x F_{+|x}(p(+|x,\lambda) - p(-|x,\lambda))\right].
\end{align*}
Using $p(-|x,\lambda)  = 1 - p(+|x,\lambda)$ and considering extremal conditional probabilities, we obtain
\begin{align*}
     V_{\mathcal L}(\mathbf{s}.\mathbf{F}) &\leq \sup_{\boldsymbol{\epsilon} \in \{\pm 1\}^g}\left\| \sum_{x \in [g]} s_x \epsilon_x F_{+|x} \right\|_\infty \\
     &=  \sup_{\boldsymbol{\epsilon} \in \{\pm 1\}^g}\left\| \left(\sum_{x \in [g]} s_x \epsilon_x F_{+|x}\right)^2 \right\|^{\frac{1}{2}}_\infty \\
     & = \left(\sum_{x \in [g]} s_x^2\right)^{\frac{1}{2}} = \|\mathbf{s}\|_2,
\end{align*}
since by the choice of $\mathbf{F}$
\begin{equation*}
    \left(\sum_{x \in [g]} s_x \epsilon_x F_{+|x}\right)^2 = \left(\sum_{x \in [g]}s_x^2\right) I.
\end{equation*}
Now, let $\sigma_{\pm|x} = \frac{1}{2d}(I \pm F_{+|x})$. It is easy to verify that this is an assemblage. Then,
\begin{align*}
    V_{\mathcal Q}(\mathbf{s}.\mathbf{F}) &\geq \frac{1}{d} \sum_{x \in [g]} s_x \mathrm{Tr}[F_{+|x}^2] \\
    &= \sum_{x \in [g]} s_x = \|\mathbf{s}\|_1,
\end{align*}
using that the $F_{+|x}$ square to the identity. Therefore, for any $\mathbf{s} \in  \Sigma(g,d)  = \Delta_\square(g,d)$, we have $V_{\mathcal L}(\mathbf{s}/\|\mathbf{s}\|_2.\mathbf{F}) \leq 1$ and thus, since $\mathbf{F}$ is unbiased, $V_{\mathcal Q}(\mathbf{s}.\mathbf{s}\|\mathbf{s}\|_2.\mathbf{F}) \leq 1$. This implies that $\big\|\mathbf{s}.\mathbf{s}/\|\mathbf{s}\|_2\big\|_1 = \|\mathbf{s}\|_2 \leq 1$, and the conclusion follows.
\end{proof}
This is a simplified version of the proof of Theorem 2 in \cite{marciniak2015unbounded}, although a small mistake led the authors to obtain $V(\mathbf{F}) \geq \sqrt{g/2}$ instead of the correct value $V(\mathbf{F}) \geq \sqrt{g}$. Another way to obtain the $d$-independent bounds of this section would have been to use the results gathered in \cite{davidson2016dilations, bluhm2018joint} and to use Proposition \ref{prop:Deltas-are-equal}.

\section{Discussion} \label{sec:discussion}

\subsection{Quantum steering}
In this section, we compare our work to previous results on the violation of steering inequalities. While previous results such as \cite{marciniak2015unbounded, Yin2015} give lower bounds on the violation of steering inequalities, our connection to inclusion constants gives upper bounds on $V(\mathbf{F})$ (see Proposition \ref{prop:sigma-is-delta}). Moreover, the main result of \cite{marciniak2015unbounded, Yin2015} is that $V(\mathbf{F})$ is unbounded for increasing dimension $d$ and number of measurements on Alice's side $g$, whereas we are interested in the possible violations for \emph{fixed} $d$ and $g$. This explains the apparent contradiction between the unbounded violations obtained in \cite{marciniak2015unbounded} by letting $g$, $d \to \infty$ and the explicit upper bounds obtained in this work for fixed $g$, $d$.

Our connection to free spectrahedra, established in Theorem \ref{thm:steering-equals-inclusion}, allows us to study in particular the regimes $d \gg g$ and $d \ll g$. In terms of $g$-independent bounds, we find that $\gamma^0_{g,d} \leq 1/\tau_*(d)$, where 
	\begin{equation} \label{eq:discussion-sigma}
		\tau_*(d) = 4^{-n}\binom{2n}{n}, \qquad \mathrm{with~} n = \lfloor d/2 \rfloor
	\end{equation}
(see Theorem \ref{thm:g-independent-bounds}). Asymptotically, $\tau_*(d)$ behaves as $\sqrt{2/(\pi d)}$. Such bounds are not considered in previous works such as \cite{marciniak2015unbounded, Yin2015}, since the authors of those papers are only interested in the dependence on the number of measurements on Alice's side $g$. The proof mainly consists of adapting the techniques in \cite{ben-tal2002tractable, helton2019dilations} to the complex setting, which is not straightforward. Moreover, Proposition \ref{prop:optimality-of-constant} shows how to construct unbiased linear steering inequalities which are arbitrarily close to achieving the bound $1/\tau_*(d)$; thereby showing that the bound is optimal. Note that the dimension dependence found in Theorem 2 of \cite{marciniak2015unbounded} is quite weak compared to our bound, since it only gives $V(\mathbf{F}) \geq O(\sqrt{\log{d}})$, which is exponentially weaker.

In terms of bounds independent of $d$, we find in Proposition \ref{prop:d-independent-bounds} that 
\begin{equation}\label{eq:quarter-circle}
    \Sigma(g,d) \supseteq \mathrm{QC}_g:=\left\{\mathbf{s} \in  [0,1]^g~:~\sum_{i = 1}^g s_i^2 \leq 1\right\}
\end{equation}
for any $d \in \mathbb N$, with equality for $d \geq 2^{\lceil (g-1)/2 \rceil}$. In particular, this means that $\gamma^0_{g,d} \leq \sqrt{g}$ with equality for  $d \geq 2^{\lceil (g-1)/2 \rceil}$. The result follows essentially from the results in \cite{passer2018minimal}. It shows that the $g$-dependence of $V(\mathbf{F})$ found in Theorem 2 of \cite{marciniak2015unbounded} is optimal and gives an explicit construction of an optimal steering inequality, based on the Pauli matrices. We note that for $g \gg d$, $	\tau_*(d) \not \in \mathrm{QC}_g$, such that we can identify a regime in which $\Sigma(g,d) \supsetneq \mathrm{QC}_g$. Figure \ref{fig:phase-diagram} collects the results obtained in this work concerning the set of steering constants $\Sigma(g,d)$.

\begin{figure}
    \centering
    \includegraphics{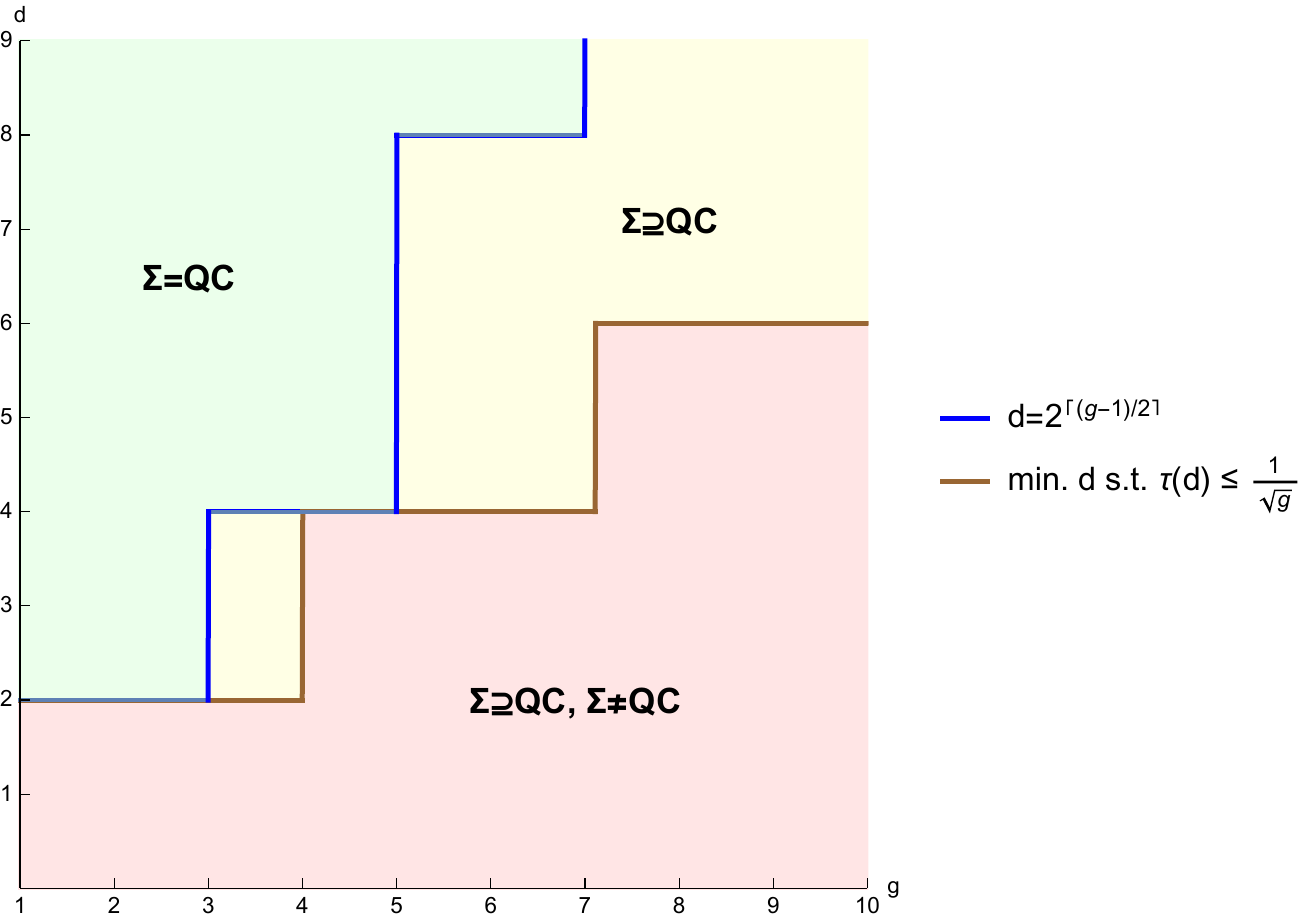}
    \caption{Bounds on the set $\Sigma(g,d)$. The set $\mathrm{QC}_g$ is defined as in \eqref{eq:quarter-circle}. Note that the blue line belongs to the green region in which $\Sigma(g,d) = \mathrm{QC}_g$, whereas the brown curve does not belong to the red region in which $\Sigma(g,d) \supsetneq \mathrm{QC}_g$}.
    \label{fig:phase-diagram}
\end{figure}

While we have focused on assemblages with $k_x=2$ for all $x \in [g]$, the connection to free spectrahedra also holds for arbitrary steering settings $\mathbf k$. The dichotomic case has however the advantage that the free spectrahedra of interest are well-studied.

\subsection{Joint measurability} 
We conclude this section with a discussion of the implications of the present work for the compatibility of dichotomoic measurements. Recall the definition of the compatibility regions studied in \cite{busch2013comparing, bluhm2018joint, bluhm2020compatibility, Bluhm2020GPT} in the dichotomic case:
\begin{align*}
    &\Gamma(g, d) :=\\& \left\{\mathbf{s} \in  [0,1]^g~:~ \mathrm{POVMs~}(E_{a|x}^\prime)_{a \in [2]}\mathrm{~are~compatible~}\forall \mathrm{~POVMs~}(E_{a|x})_{a \in [2]} \subset (\mathcal M_d^+)^{2}\right\},
\end{align*}
where 
\begin{equation*}
    E_{a|x}^\prime := s_x E_{a|x} + (1-s_x) \frac{I_d}{2}.
\end{equation*}
As proven in \cite{bluhm2018joint} (see also \eqref{eq:Gamma-is-jewel}), it holds that 
\begin{equation*}
     \Gamma(g, d) = \Delta_{\diamond}(g,d).
\end{equation*}
Since by Proposition \ref{prop:Deltas-are-equal} $\Delta_{\diamond}(g,d) = \Delta_{\square}(g,d)$, the inclusion constants for the matrix cube derived in Section \ref{sec:cube-csts} lead to new bounds on the amount of incompatibility available for a fixed dimension and fixed number of measurements as represented by $\Gamma(g, d)$. Building on \cite{helton2019dilations}, Proposition 7.2 in \cite{bluhm2018joint} shows that $(2d)^{-1}(1, \ldots, 1) \in \Gamma(g, d)$. Theorem \ref{thm:g-independent-bounds} improves upon this and shows that  $\tau_*(d)(1, \ldots, 1) \in \Gamma(g, d)$, where again $\tau_*(d)$ as in \eqref{eq:discussion-sigma}. As asymptotically $\tau_*(d)$ behaves as $\sqrt{2/(\pi d)}$, this improves over the previous bound, since the behavior is now $\Theta(d^{-1/2})$ instead of $\Theta(d^{-1})$. Thus, we have derived better $g$-independent bounds than in \cite{bluhm2018joint}, whereas the $d$-independent bounds we find are the same. Finally, in the case of qubits, i.e.\ $d=2$, we find that $\tau_*(2) = 1/2$. This recovers a result from \cite{Bluhm2020GPT}, which has been obtained using completely different techniques, exploiting a connection to $1$-summing norms of $\ell_2$ Banach spaces.

\textbf{Acknowledgements:}
The authors would like to thank Franck Barthe for providing us with the proof of Lemma \ref{lem:minimum-signed-chi2} in the even case, as well as MathOverflow users ``Fedor Petrov'' and ``Steve'' for providing simpler, more conceptual proofs for parts of Lemma \ref{lem:minimum-signed-chi2}. A.B.~acknowledges financial support from the VILLUM FONDEN via the QMATH Centre of Excellence (Grant No.10059) and the QuantERA ERA-NET Cofund in Quantum Technologies implemented within the European Union’s Horizon 2020 Programme (QuantAlgo project) via the Innovation Fund Denmark. I.N.~was supported by the ANR project ``ESQuisses'' (grant number ANR-20-CE47-0014-01).

\bibliographystyle{plainnat}
\bibliography{spectralit}
\end{document}